\documentclass{llncs}
\usepackage{xcolor}
\definecolor{cream}{RGB}{255,249,235}
\usepackage{tikz}
\usetikzlibrary{arrows.meta,positioning,shapes.geometric,fit,calc}
\usepackage{algorithm}
\usepackage{algpseudocodex}
\usepackage{amssymb}
\usepackage{enumerate}
\usepackage{amsmath}
\usepackage{academicons}
\usepackage{xcolor}
\usepackage{hyperref}
\usepackage{xcolor}
\usepackage{soul}
\usepackage{graphicx}
\usepackage{lineno}
\usepackage{booktabs,longtable,tabularx}
\usepackage{tikz}
\usetikzlibrary{arrows.meta,positioning}
\usetikzlibrary{positioning,fit}
\usepackage{booktabs}
\usepackage{tabularx}
\usepackage{array}
\usepackage{multirow}
\usepackage{listings}
\usepackage{tabularx}
\usepackage[a4paper, total={7in, 9in}]{geometry}
\usepackage{array}
\usepackage{algorithm}
\usepackage{algpseudocodex}

\makeatletter
\def\algbackskip{\hskip-\ALG@thistlm}
\makeatother
\AtBeginDocument{%
 \providecommand\BibTeX{{%
   \normalfont B\kern-0.5em{\scshape i\kern-0.25em b}\kern-0.8em\TeX}}}
\DeclareUnicodeCharacter{2212}{\ensuremath{-}}   
\begin{document}
\title{DPLAN: Minimal Connectivity to Floorplan Generation}
 \author{Rohit Lohani, Krishnendra Shekhawat}
\author{
Rohit Lohani$^{1}$\thanks{Corresponding author: Rohit Lohani, p20210045@pilani.bits-pilani.ac.in} \and
Krishnendra Shekhawat$^{1}$
}
\institute{BITS Pilani, Department of Mathematics, Pilani Campus, India}
\maketitle
\begin{abstract}
Automated floor plan generation is an important problem in computational architectural design. The goal is to construct a floor plan from user-defined room numbers and door requirements. The user specifies which rooms must share a door and which rooms must not be adjacent. However, these requirements do not determine the exact placement or shape of the rooms. The task is therefore to arrange the rooms in a single floor plan so that all required door connections are satisfied and no rooms overlap. To address this problem, we propose \emph{Door Connectivity to Floor Plan Generation (DPLAN)}, a graph-based prototype that generates floor plans from door and non-adjacency constraints. 
The framework operates in three stages.
\begin{enumerate}
    \item \textbf{Connectivity adjustment:} 
    The user-defined graph is examined. If it is disconnected, additional edges are added to connect its components.
\item \textbf{Graph construction:} 
    From the connected graph, a bi-connected plane triangulation is constructed to ensure the existence of a floor plan without overlapping or empty spaces.
\item \textbf{Floor plan generation:} 
    The triangulated graph is transformed into floor plans using two modes.
    \begin{itemize}
    \item For \textbf{rectangular floor plans (RFPs)}, separating triangles are removed by modifying edges without adding new vertices, so that no extra rooms are created.
  \item For \textbf{orthogonal floor plans (OFPs)}, separating triangles are removed by introducing additional vertices, which allows rectilinear room shapes.
    \end{itemize}
\end{enumerate}
By enforcing both door and non-adjacency requirements, the framework generates floor plans that satisfy the given constraints. 
Our proposed method is based on graph construction, plane triangulation, and systematic graph modification. The proposed method is implemented in Python and includes a prototype that allows interactive input of constraints and visualization of the generated floor plans.
Currently, the framework supports rectangular plot boundaries. 
Future work includes support for non-rectangular plots, dimension-based scaling, and circulation modeling to further extend its practical applicability.
\end{abstract}  
\section{Introduction}
\begin{figure}
\centering
\resizebox{\linewidth}{!}{%
\begin{tikzpicture}[
  node distance=25mm and 20mm,
  box/.style={
    draw, rounded corners,
    fill=cream,
    align=left,
    inner sep=6pt,
    font=\small,
    text width=6.6cm
  },
  widebox/.style={
    draw, rounded corners,
    align=left,
    fill=cream,
    inner sep=6pt,
    font=\small,
    text width=13.6cm
  },
  decision/.style={
    draw, diamond,
    fill=cream,
    aspect=2.2,
    align=center,
    inner sep=2pt,
    font=\small
  },
  arrow/.style={-Latex, thick}
]
\node[widebox] (start) {%
\textbf{Start: User Inputs}\\
(1) Plane graph $G=(V,E)$\\
(2) forbidden set $N$\\
(3) Rectangular boundary\\
(4) Choose mode: RFP or OFP
};

\node[decision, below=of start] (conn) {Is $G$\\connected?};

\node[box, left=of conn, xshift=-10mm] (stage1) {%
\textbf{Stage 1: Connectivity augmentation}\\
\textbf{Algorithm \ref{OF-CON}} $\rightarrow G_C$\\
Connect components using outer-face vertices;\\
reject candidate edges in $N$ when feasible;\\
fallback may relax $N$
};

\node[box, below=of conn] (stage2) {%
\textbf{Stage 2: Bi-connectivity augmentation}\\
\textbf{Algorithm \ref{BY}} $\rightarrow G_B$\\
Eliminate articulation points by connecting blocks;\\
respect $N$ when feasible; fallback may relax $N$
};

\node[widebox, below=of stage2] (stage3) {%
\textbf{Stage 3: Plane triangulation}\\
\textbf{Algorithm \ref{TRI}} $\rightarrow G_T$\\
Triangulate non-triangular interior faces via ear-clipping;\\
\texttt{IsValidDiagonal} rejects $(a,c)\in N$;\\
Fallback may insert an edge even if forbidden
};

\node[decision, below=of stage3] (mode) {Mode?};

\node[box, left=of mode, xshift=-10mm] (stage4rfp) {%
\textbf{Stage 4 (RFP): Remove separating triangles}\\
\textbf{Algorithm \ref{STR}} $\rightarrow G_F$ (PTPG)\\
Try: delete exterior ST edge; else replace an interior edge and $\#$ST decreases; respect $N$ when feasible;\\
fallback may relax $N$
};

\node[box, below=of stage4rfp] (stage5rfp) {%
\textbf{Stage 5 (RFP): Rectangular floor plans}\\
Input: PTPG $G_F$\\
Apply rectangular-dual method (Koźmiński et al.~\cite{kozminski1985rectangular})\\
Output: multiple dimensionless RFPs (UFPs)
};

\node[box, right=of mode, xshift=10mm] (stage4ofp) {%
\textbf{Stage 4 (OFP): Allow rectilinear rooms}\\
Start from $G_T$ (PTG may contain separating triangles)
};

\node[box, below=of stage4ofp] (stage5ofp) {%
\textbf{Stage 5 (OFP): OFP synthesis}\\
Add auxiliary vertices to eliminate separating triangles;\\
triangulate as needed $\rightarrow$ ST-free modified graph;\\
build an RFP for modified graph;\\
merge auxiliary regions into neighbors\\
Output: multiple OFPs
};

\coordinate (midstage5) at ($(stage5rfp.south)!0.5!(stage5ofp.south)$);
\node[widebox] (gui) at ($(midstage5)+(0,-18mm)$) {%
\textbf{Our Prototype:}\\
Visualize user-edges vs added edges.
};

\draw[arrow] (start) -- (conn);

\draw[arrow] (conn.west) -- node[above, font=\small]{No} (stage1.east);
\draw[arrow] (stage1) |- (stage2);

\draw[arrow] (conn) -- node[right, font=\small]{Yes} (stage2);

\draw[arrow] (stage2) -- (stage3);
\draw[arrow] (stage3) -- (mode);

\draw[arrow] (mode.west) -- node[above, font=\small]{RFP} (stage4rfp.east);
\draw[arrow] (stage4rfp) -- (stage5rfp);

\draw[arrow] (mode.east) -- node[above, font=\small]{OFP} (stage4ofp.west);
\draw[arrow] (stage4ofp) -- (stage5ofp);

\draw[arrow] (stage5rfp) |- (gui);
\draw[arrow] (stage5ofp) |- (gui);

\end{tikzpicture}
}
\caption{Flowchart of the DPLAN pipeline from minimal door-connectivity input (with non-adjacency constraints) to RFP/OFP generation.}
\label{FLOW}
\end{figure}
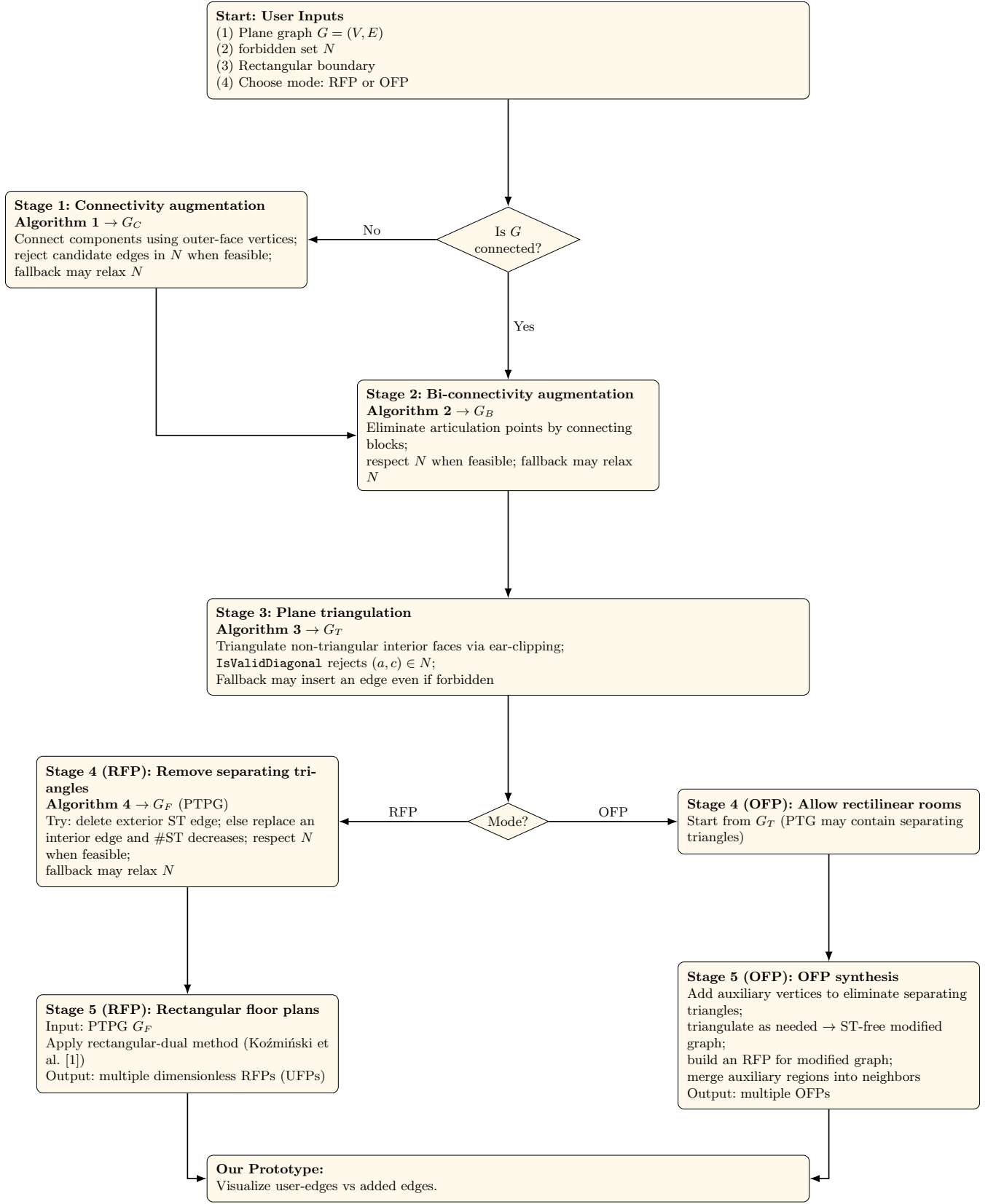
\begin{figure}
\centering
    \includegraphics[width = 0.95 \textwidth]{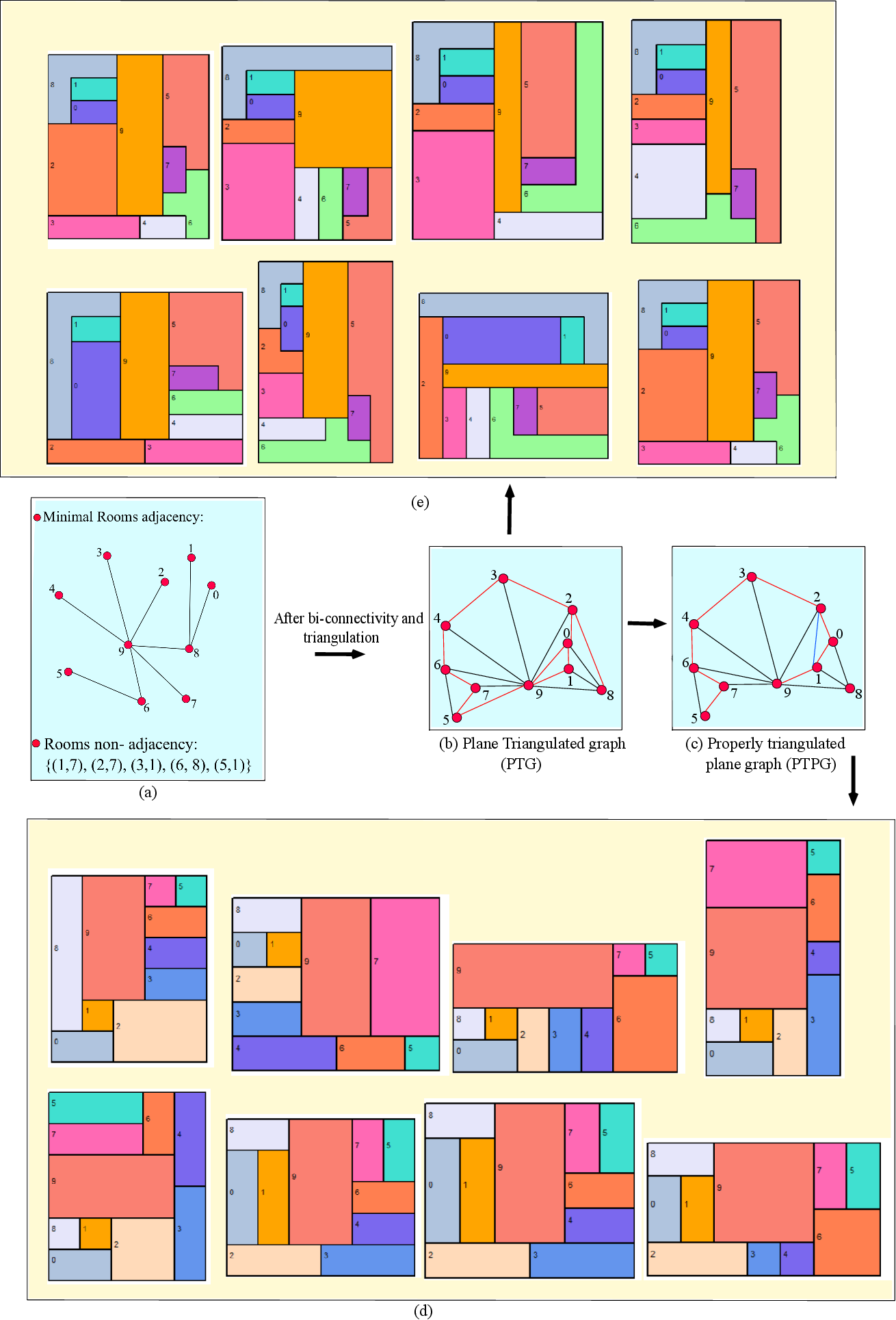}
    \caption{(a) Minimal input graph with non-adjacency constraints. (b) A plane triangulated graph constructed using the proposed approach. (c) A properly triangulated plane graph obtained from the proposed approach. (d) Orthogonal floor plans generated from the resulting PTG. (e) Rectangular floor plans generated from the resulting PTPG.}
    \label{FF-1}
\end{figure}

Floor planning is a fundamental problem in architectural design because the arrangement of rooms determines access between spaces and the connections between rooms. A floor plan must satisfy several requirements, such as which rooms should be adjacent, which rooms should remain separate, and how the design fits within the given plot boundary. 
These requirements influence one another, which makes the design process difficult. As the number of rooms increases, the number of possible arrangements grows rapidly. For this reason, manual exploration becomes time-consuming, and simple computational search methods become inefficient.\\
Over the years, many automated floor-planning methods have been developed. These include optimization-based techniques, shape grammar approaches, graph-based models, and machine learning methods (discussed in section \ref{Lite}). 
Among them, graph-based approaches are effective because they represent rooms as vertices and spatial relationships as edges. 
This representation clearly shows the connections and separations between rooms. Graph-based models have also been used in areas such as VLSI floorplanning and computational architecture, where connectivity and planarity are important. However, many existing systems require users to specify complete adjacency information in advance or depend on heuristic optimization methods. This reduces flexibility and limits the range of floor plans that can be generated. In practical design settings, architects often specify only essential requirements, such as which rooms must be connected, which must remain separate, the shapes of rooms, and the boundary of the floor plan. Ensuring that both adjacency and non-adjacency requirements are satisfied in the final floor plan in a systematic way is challenging. Although modern computer-aided design tools and data-driven methods (see Section~\ref{Lite}) can generate initial layouts, they provide limited control over the underlying structure. It is often unclear how properties such as bi-connectivity, triangulation, and the handling of separating triangles are ensured. \\
In this paper, we present \emph{Door Connectivity to Floor Plan Generation} (DPLAN), a graph-based framework for automated floor plan generation (see Figure \ref{FLOW}). The system builds floor plans from essential adjacency and non-adjacency requirements derived from door connections, 
without requiring complete room-to-room specifications (see Figure~\ref{FF-1}). The framework first ensures that the input graph is connected. It then converts the graph into a bi-connected plane-triangulated form. From this representation, multiple floor plans are generated that satisfy the given constraints. 
The resulting floor plans remain valid while allowing different room arrangements.\\
The current implementation supports dimensionless floor plans within rectangular plot boundaries. This provides a basis for studying the structural feasibility and variation in floor-plan design. Future work includes support for non-rectangular plots, room dimensions, and circulation modeling to extend the method toward practical architectural applications.
By combining graph-based modeling with interactive software support, 
DPLAN offers a clear and extensible approach to automated floor plan generation.\\
The paper is organized as follows. 
Section~\ref{Preliminaries} introduces the main definitions and notation used in this work. Section~\ref{Lite} reviews related literature. 
Sections~\ref{Gap} and~\ref{Our} discuss existing approaches, identify research gaps, and position our contribution. Section~\ref{Meth} describes the proposed methodology, including the step-by-step process from graph construction to floor plan generation, along with the associated algorithms. Section~\ref{Apply} explains the use of the DPLAN interface. 
Sections~\ref{conclusion} and~\ref{FE} discuss limitations and directions for future work. Finally, Section~\ref{App} provides an appendix that includes runtime analysis, comparisons with existing methods, the scope of the approach, and computational complexity and correctness analyses of the implemented algorithms.
\begin{figure}
     \centering
  \includegraphics[scale=0.85]{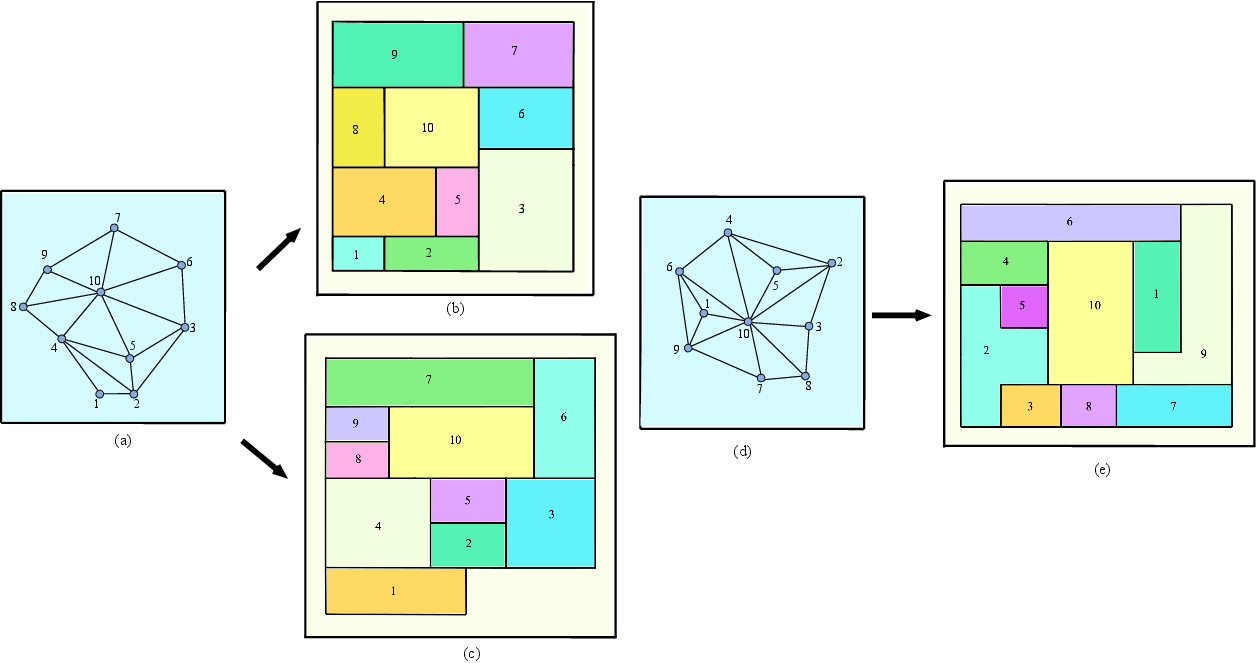}
     \caption{(b) Rectangular floor plan based on the graph in (a); (c) non-rectangular floor plan derived from the same graph; (e) orthogonal floor plan generated from the graph in (d).}
     \label{Rule-1}
     \end{figure}
\section{Terminologies}\label{Preliminaries}
This section introduces the basic concepts and terminology used in this study. 
These definitions provide a clear foundation for representing and analyzing floor plans using graph structures and spatial relationships.
\begin{enumerate}[(i).]

\item \textbf{Floor Plan \cite{rinsma_1988}:}
A \emph{floor plan} ($F$) consists of an outer polygonal region that is internally divided into $n$ non-overlapping sub-regions, each representing an individual \emph{room}. The arrangement of these rooms follows the connectivity prescribed by a graph on $n$ vertices, where each vertex corresponds to a room. The polygon enclosing the entire subdivision is referred to as the \emph{floor plan boundary}, and the straight line segments that separate neighboring rooms form the \emph{walls}. Two rooms are regarded as adjacent precisely when they meet along a wall segment of positive length.

Floor plans can be classified according to the geometry of their outer boundary and the shapes of the rooms they contain. The most common categories are described below:

\begin{enumerate}[(a).]

\item \textbf{Rectangular Floor Plan (RFP):}  
A floor plan is termed a \emph{rectangular floor plan} when both the enclosing plot and all interior spaces are bounded by rectangles. In such layouts, every room admits four orthogonal sides, and the overall configuration fits entirely within a rectangular boundary (see Figure~\ref{Rule-1}a).

\item \textbf{Orthogonal Floor Plan (OFP):}  
A floor plan is described as an \emph{orthogonal floor plan} if its outer boundary is rectangular, while allowing some rooms to assume orthogonal shapes that are not strictly rectangular, including configurations such as \emph{L}-, \emph{T}-, or \emph{C}-shaped regions. An illustration of such rooms is shown in Figure~\ref{Rule-1}e (rooms 9 and 2).

\item \textbf{Non-rectangular Floor Plan (NRFP):}  
A floor plan is categorized as a \emph{non-rectangular floor plan} when the shape of its outer boundary is not rectangular, independent of whether the internal rooms themselves are rectangular or non-rectangular. An example is provided in Figure~\ref{Rule-1}c.

\end{enumerate}

\item \textbf{Graph Representations of Floor Plans:}
A floor plan can be naturally modeled using graph-theoretic structures. In such representations, each vertex corresponds to a room, and an edge between two vertices indicates that the corresponding rooms share a common wall. Graphs constructed in this manner are known as \emph{adjacency graphs}. Modeling floor plans as adjacency graphs enables the application of established graph algorithms for analysis, optimization, and floor plan synthesis.

In addition to adjacency graphs, floor plans are often studied using \emph{dual graphs}. Given a planar embedding of a graph, its dual is obtained by associating each face of the original (primal) graph with a vertex in the dual graph. Two vertices in the dual graph are connected by an edge if and only if their corresponding faces in the primal graph share a common boundary edge.

Within the scope of this work, adjacency graphs for RFPs and NRFPs consist of vertices representing rectangular rooms and edges representing shared walls between them. In contrast, graphs derived from OFPs may include vertices associated with both non-rectangular and rectangular rooms with orthogonally shaped boundaries, while edges consistently represent adjacency through common walls.\\
\item \textbf{Plane Graphs and Their Variants \cite{bhasker1987linear}:}
A \emph{plane graph} is a graph together with a fixed drawing in the plane in which edges are represented as non-intersecting curves, meeting only at shared endpoints. This drawing induces a decomposition of the plane into connected regions known as \emph{faces}. The unique unbounded region is called the \emph{outer} (or exterior) face, while all bounded regions are referred to as \emph{inner} (or interior) faces. A vertex that lies on the boundary of the outer face is termed an \emph{exterior vertex}; all remaining vertices are classified as \emph{interior vertices}.

Certain subclasses of plane graphs play a central role in floor plan generation and are described below:

\begin{enumerate}[(i).]

\item \textbf{Plane Triangulated Graph (PTG):}  
A plane graph is considered \emph{triangulated} when every interior face is enclosed by exactly three edges, whereas the outer face may consist of three or more edges (see Figure~\ref{Rule-1}d). A triangular cycle whose interior contains at least one vertex is referred to as a \emph{separating triangle}, as it divides the graph into distinct regions (see Figure~\ref{Rule-1}d, separating triangle $(6,9,10)$).

\item \textbf{Properly Triangulated Plane Graph (PTPG):}  
A \emph{properly triangulated plane graph} is a connected plane graph in which all interior faces are triangular, and no separating triangles occur (see Figure~\ref{Rule-1}b). The absence of separating triangles imposes a stronger structural condition, which is useful for algorithmic construction of floor plans.

\end{enumerate}

\end{enumerate}

\section{Literature Survey}
\label{Lite}

Automated floor plan generation has developed along two main research directions. 
The first direction is constraint-based. 
It uses graph models or rule-based procedures to explicitly represent requirements such as connectivity, planarity, and other structural conditions needed for floor plan construction. 
The second direction is data-driven. 
It learns geometric patterns and design regularities from large collections of existing floor plans, often using generative models that convert a graph specification into a geometric layout.\\
In constraint-based methods, the floor plan is viewed as a partition of the building boundary into rooms, and the design requirements are encoded as an adjacency or connectivity graph. 
These methods are interpretable and can provide guarantees for the constraints they explicitly model. 
However, they often assume that a well-structured graph is already given. 
The process of converting sparse design requirements into such a structured graph is not always addressed directly.\\
In contrast, learning-based methods can generate visually diverse and realistic layouts. 
However, constraints such as strict non-adjacency, triangulation, or bi-connectivity are usually encouraged through training objectives rather than guaranteed by construction.

\subsection{Graph-theoretic approaches}
Table~\ref{tab:rule_methods} presents representative constraint-based approaches that start from an explicit adjacency graph and generate candidate floor plans through graph transformations or optimization procedures. 
These methods work effectively when the input graph is already consistent and sufficiently detailed. 
However, they generally do not address the problem of minimally modifying an input graph while enforcing explicit non-adjacency requirements.

\subsection{Learning-based and hybrid approaches}
Table~\ref{learning_methods} presents representative learning-based approaches that generate room geometries from a given graph and boundary condition. 
These methods can produce realistic and diverse floor plans. 
However, structural constraints (such as adjacency, triangulation, connectivity) are usually encouraged through training objectives rather than guaranteed by construction, and the input graph is typically assumed to be internally consistent.

\begingroup
\small
\setlength{\tabcolsep}{6pt}
\renewcommand{\arraystretch}{1.35}

\begin{longtable}{p{2.0cm} p{0.8cm} p{3.2cm} p{5.0cm} p{5.2cm}}
\caption{Graph-theoretic and rule-based methods for floor plan generation.}
\label{tab:rule_methods}\\

\toprule
\textbf{Work} & \textbf{Year} & \textbf{Primary input} & \textbf{Main idea / representation} & \textbf{Output and remarks} \\
\midrule
\endfirsthead

\toprule
\textbf{Work} & \textbf{Year} & \textbf{Primary input} & \textbf{Main idea / representation} & \textbf{Output and remarks} \\
\midrule
\endhead

\midrule
\multicolumn{5}{r}{\small\textit{Continued on next page}}\\
\endfoot

\bottomrule
\endlastfoot

Ko\'zmi\'nski \& Kinnen \cite{kozminski1985rectangular} &
1985 &
Planar graph under rectangular-dual conditions &
Characterizes when an adjacency graph admits a rectangular dual and gives a constructive method &
Strong guarantees, but assumes a highly structured planar input instead of building one from sparse constraints. \\

Rinsma \cite{rinsma1988rectangular} &
1988 &
Graph-theoretic floor plan existence setting &
Provides existence-style results under stated combinatorial conditions &
Mainly theoretical; not presented as a complete generation pipeline. \\

Lai \& Leinwand \cite{lai1988algorithms} &
1990 &
Plane graph (embedded) &
Reduces rectangular-dual construction to a bipartite matching problem on a derived graph &
Gives an algorithm when the input is already planar/embedded and satisfies the required conditions. \\

M. Kurowski \cite{kurowski2003simple} &
2003 &
Plane near-triangulation / suitable planar graph &
Constructs a dissection-like floor plan from a suitable planar graph using a simple procedure &
Efficient construction, but does not target rich architectural constraints (e.g., explicit forbidden adjacencies). \\

Liao \textit{et al.} \cite{liao2003compact} &
2003 &
Planar embedding via orderly spanning tree &
Uses planar encoding to build compact rectilinear drawings/dissections &
Focused on planar encoding/drawing; not designed for constraint-heavy floor plan briefs. \\

Marson \& Musse \cite{marson2010automatic} &
2010 &
Room list and rough size/grouping intent &
Recursive partitioning (treemap-like subdivision) followed by adjustments &
Fast and interactive, but strict adherence to the given adjacency/non-adjacency is not the main guarantee. \\

Zhang \textit{et al.} \cite{zhang2011improved} &
2011 &
Plane triangulation &
Studies transformations that improve triangulation structure for downstream constructions &
Useful when a triangulation is already available; it does not focus on repairing sparse/disconnected briefs. \\

Mirahmadi \& Shami \cite{mirahmadi2012novel} &
2012 &
Procedural rules + room program &
Rule-based subdivision with corridor-driven heuristics and optimization &
Produces plausible plans, but relies on heuristics rather than certified graph conditions. \\

Hao Hua \cite{hua2016irregular} &
2016 &
Template/shape cues + relational constraints &
Matches a target template/topology and uses search/optimization for feasibility &
Handles irregular regions well; feasibility is mainly achieved by optimization. \\

Wang \textit{et al.} \cite{wang2018customization} &
2018 &
Existing floor plan (or derived structure) + edits &
Edits/repairs a plan by graph extraction and graph transformations &
Effective when a coherent seed plan exists; less focused on generating from scratch under sparse constraints. \\

Upasani \textit{et al.} \cite{upasani2020automated} &
2020 &
A valid dimensionless arrangement + geometric bounds &
First fixes a valid arrangement, then assigns dimensions via optimization &
Produces dimensioned plans, assuming a valid topology/arrangement is already given. \\

Shekhawat \textit{et al.} \cite{SHEKHAWAT2021103718} &
2021 &
Adjacency graph (+ optional boundary/dimensions) &
Combines graph construction with optimization/dimensioning for floor plan generation &
Algorithmic workflow, but assumes the adjacency specification is already meaningful and complete enough. \\

Shekhawat \textit{et al.} \cite{shekhawat2023automated} &
2023 &
Input graph + room-shape categories &
Generates plans allowing specified orthogonal/non-rectangular room shapes within constraints &
Supports richer room shapes; feasibility depends on the structural graph and the allowed shape family. \\

Bisht \textit{et al.} \cite{https://doi.org/10.1111/cgf.14451} &
2022 &
Adjacency graph (+ optional dimensions) &
Enumerates topologically distinct floor plan/candidates consistent with a given graph, then dimensions them &
Good for exploring alternatives once the adjacency graph is fixed; less focused on minimal graph repair/augmentation. \\

Shiksha \textit{et al.} \cite{SHIKSHA2026104506} &
2025 &
Rule-level brief (required/optional adjacency, forbidden contacts, circulation cues) &
Converts a brief into candidate design graphs using rules, then synthesizes plans &
Expressive for brief-to-graph translation; emphasis is not on minimal augmentation to a certified backbone under all constraints. \\

\end{longtable}
\endgroup
\begingroup
\small
\setlength{\tabcolsep}{6pt}
\renewcommand{\arraystretch}{1.25}

\begin{longtable}{p{2.0cm} p{0.8cm} p{3.2cm} p{5.0cm} p{5.2cm}}
\caption{Learning-based (data-driven) methods for floor plan generation and reconstruction.}
\label{learning_methods}\\

\toprule
\textbf{Work} & \textbf{Year} & \textbf{Primary input} & \textbf{Main idea / representation} & \textbf{Output and remarks} \\
\midrule
\endfirsthead

\toprule
\textbf{Work} & \textbf{Year} & \textbf{Primary input} & \textbf{Main idea / representation} & \textbf{Output and remarks} \\
\midrule
\endhead

\midrule
\multicolumn{5}{r}{\small\textit{Continued on next page}}\\
\endfoot

\bottomrule
\endlastfoot

C. Liu \cite{liu2018floornet} &
2018 &
RGBD video / 3D scans &
Multi-branch network predicts floor plan geometry/semantics, with an IP step for vector output &
Targets floor plan reconstruction from scans rather than brief-based generation. \\

W. Wu \textit{et al.} \cite{wu2019data} &
2019 &
Boundary + dataset priors (optionally, room program) &
Two-stage learned synthesis: allocate rooms, then generate/refine walls within the boundary &
Fast and plausible, but strict feasibility may require post-processing. \\

J. Chen \textit{et al.} \cite{Chen_2019_ICCV} &
2019 &
RGBD evidence (top-view cues) &
Room-wise optimization guided by learned cues and consistency terms &
Focused on reconstruction quality, not free-form generation from constraints. \\

Hu \textit{et al.} \cite{10.1145/3386569.3392391} &
2020 &
Constraint graph + boundary &
Graph2Plan: generates a floor plan from a graph and boundary using learned decoding stages &
Can follow sparse constraints, but hard constraints often need explicit validation/repair. \\

Nauata \textit{et al.} \cite{nauata2020house} &
2020 &
Adjacency graph (rooms/types/contacts) &
House-GAN: graph-conditioned GAN generates room boxes consistent with the input graph &
Diverse generation; hard constraints may be violated without explicit checking. \\

Nauata \textit{et al.} \cite{nauata2021house} &
2021 &
Graph + initial plan for refinement &
House-GAN++: iterative refinement improves a generated plan under the same graph conditioning &
Cleaner outputs, but strict feasibility still typically needs validation steps. \\

S. Wang \textit{et al.} \cite{wang2021actfloor} &
2021 &
Boundary + activity/function priors &
ActFloor-GAN: activity-guided synthesis for human-centric residential plans &
Produces plausible designs; constraint correctness is mainly learned rather than guaranteed. \\

J. Sun \textit{et al.} \cite{sun2022wallplan} &
2022 &
Boundary + design constraints &
WallPlan: predicts an intermediate wall-graph, then decodes to a structured plan &
More structured than raster-only methods; still may require rule checks for CAD validity. \\

S. Shabani \textit{et al.} \cite{shabani2023housediffusion} &
2023 &
Constraint graph + vector corner/door representation &
HouseDiffusion: conditional diffusion that directly generates vector floor plans &
Stronger geometric control; performance depends on constraint encoding and training coverage. \\

Z. Han \textit{et al.} \cite{hangraph2pix} &
2024 &
Adjacency graph (+ attributes) &
Graph2pix: graph-conditioned generator produces floor plan images aligned to the graph &
Usually followed by vectorization/cleanup for CAD use. \\

Chen Ma Thi \textit{et al.} \cite{thi2024deep} &
2024 &
Boundary + user constraints &
Learns boundary partitioning under constraints, then reconstructs a coherent plan &
Practical synthesis, but formal guarantees are typically not provided. \\

Zhang \textit{et al.} \cite{zhang2024maskplan} &
2024 &
Partial user specification (incomplete constraints/attributes) &
MaskPLAN: masked generation completes missing attributes and produces a plan from partial input &
Supports interactive partial briefs; still benefits from explicit feasibility checks. \\

S. Hong \textit{et al.} \cite{hong2024cons2plan} &
2024 &
Graph, boundary, or both &
Cons2Plan: conditional diffusion generates vector floor plans under various conditions &
Vector output helps, but architectural constraints often still need validation. \\

P. Zeng \cite{zeng2024residential} &
2024 &
Multiple conditions (boundary + functional constraints) &
Diffusion-based synthesis designed to improve controllability under multiple inputs &
Good diversity/control; CAD-grade validity often requires post-processing. \\

M. Abouagour \textit{et al.} \cite{abouagour2025gflan} &
2025 &
Boundary (and door/entry context) &
Factorizes the task: first plan topology, then realize geometry with neural components &
Explicitly separates topology from geometry; still not a formal graph-certification pipeline. \\

\end{longtable}
\endgroup

\section{Gap in the Existing Research}\label{Gap}
\begin{figure}
     \centering
  \includegraphics[scale=0.85]{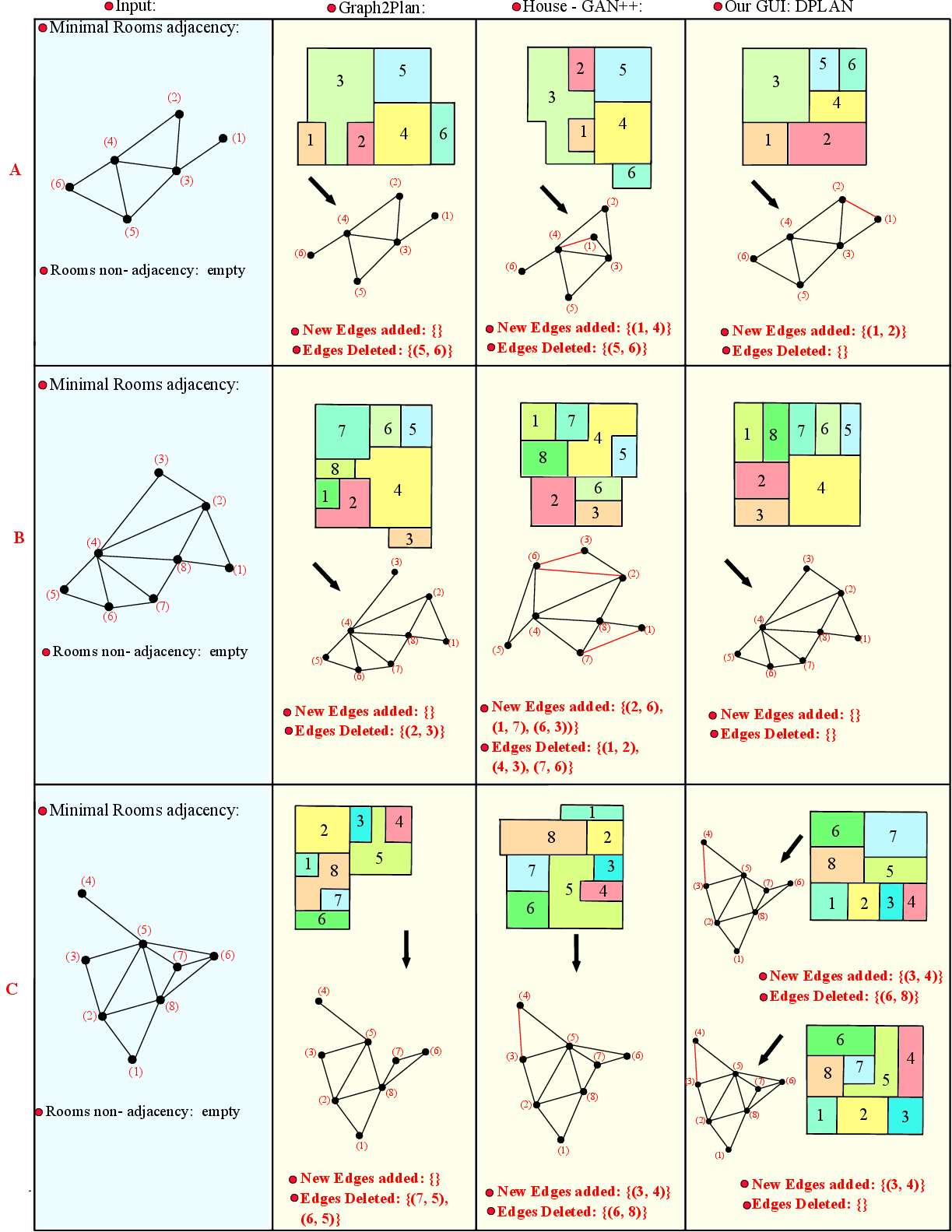}
     \caption{Visual comparison of inputs and generated outputs using our method, HouseGAN++, and Graph2Plan.}
     \label{F-10}
     \end{figure}
In the early phase of architectural design, floor plan requirements are usually incomplete. Designers often specify only a few essential room relationships and outline a rough building boundary, while leaving many other adjacency decisions open. This reflects the exploratory nature of conceptual design, in which spatial arrangements are refined step by step rather than fixed at the outset. In contrast, many automated floorplanning methods assume a much stronger and more structured input. Typically, they require a planar, connected, triangulated adjacency graph that satisfies the conditions for constructing a rectangular/orthogonal floor plan. In these approaches, connectivity between spaces is often inferred from shared boundaries, and explicit non-adjacency constraints are rarely enforced during the initial graph construction. As a result, there is a gap between the flexible, partially specified inputs used in practice and the stricter assumptions made by many computational models.\\
A large body of graph-theoretic \cite{SHEKHAWAT2021103718,https://doi.org/10.1111/cgf.14451,marson2010automatic,mirahmadi2012novel,wang2018customization,hua2016irregular} research has demonstrated how adjacency graphs can be translated into families of feasible floor plans, particularly for rectangular or orthogonal floor plans. Several studies have shown that properly triangulated plane graphs (PTPGs), or PTGs, a closely related planar graph class, serve as a necessary backbone for generating rectangular floor plans. These methods are effective once such a graph is available, but they typically assume that the user already possesses the expertise to provide or construct a valid graph. From an architectural perspective, this places a substantial burden on the designer, who must understand graph planarity, connectivity, triangulation, and the presence or absence of separating triangles to provide a usable input. How to derive such graphs algorithmically from sparse adjacency and non-adjacency information remains largely unexplored.\\
Recent years have also seen significant progress in learning-based approaches for layout/floor plan generation, including systems such as Graph2Plan\cite{10.1145/3386569.3392391}, House-GAN\cite{nauata2020house}, House-GAN++\cite{nauata2021house}, ActFloor\cite{wang2021actfloor}, and WallPlan\cite{sun2022wallplan}, Graph2pix\cite{hangraph2pix}, etc.  These methods leverage large datasets to learn spatial patterns and generate visually plausible layouts under given boundary conditions. While powerful, their behaviour is fundamentally data-driven. Constraint satisfaction is guided by learned priors or soft loss functions rather than explicit combinatorial guarantees. As a result, adjacency requirements, non-adjacency requirements, and structural properties such as triangularity or bi-connectivity may fail to hold, particularly when the input specification is sparse, inconsistent, or deviates from the training distribution (see Figure \ref{F-10}). When such violations occur, these models offer limited repair mechanisms, since they do not reason directly about the underlying door-connectivity graph (see Table~\ref{Compare}). These systems generate a floor plan conditioned on an input relationship graph, but they treat that graph as a prerequisite rather than an object to be constructed further. As a result, when the user specification is sparse, disconnected, or inconsistent, neither learning-based methods provides a mechanism to complete and certify the underlying door graph before geometric layout generation.\\
Taken together, the literature lacks a transparent, principled mechanism that bridges the gap between high-level architectural briefs and downstream floor-plan/layout generation. There is no standard pipeline that repairs incomplete or disconnected specifications, enforces connectivity and bi-connectivity under hard non-adjacency constraints, and upgrades the resulting structure into a properly triangulated plane graph or plane triangulated graph suitable for established rectangular or orthogonal floorplanning algorithms. The present work addresses this missing layer by treating the constraint graph itself as the primary design artifact. Starting from minimum adjacency and explicit non-adjacency information, we construct and certify a PTPG or PTG through a sequence of graph-theoretic augmentations, adding only the necessary edges. In doing so, door connectivity and separation requirements become explicit, verifiable, and controllable objects, rather than incidental by-products of geometric optimization or data-driven synthesis.
\begin{table}[t]
\centering
\caption{Comparison between Graph2Plan, HouseGAN++, and DPLAN}
\label{Compare}
\renewcommand{\arraystretch}{1.2}
\scriptsize
\begin{tabular}{p{2.5cm} p{3.3cm} p{3.3cm} p{3.3cm}}
\toprule
\textbf{Parameter} &
\textbf{Graph2Plan~\cite{10.1145/3386569.3392391}} &
\textbf{HouseGAN++~\cite{nauata2021house}} &
\textbf{DPLAN (Proposed Method)} \\
\midrule

Input specification
& Room connectivity graph with room types and counts, together with a boundary extracted from the dataset.
& Bubble diagram describing room relationships and a coarse program; boundary control is indirect.
& Adjacency graph with an explicit non-adjacency set and a user-defined rectangular outer boundary. \\

\midrule
Adjacency satisfaction
& Attempts to realize the input graph, but generated layouts may introduce or miss certain adjacencies.
& Encourages specified relationships, yet some edges may not be preserved due to generative sampling.
& All required adjacencies are strictly enforced, and all non-adjacent pairs remain separated throughout graph construction and in the final layout. \\

\midrule
Floorplan boundary
& Follows dataset boundaries; minor deviations may occur in generated results.
& Produces data-driven irregular boundaries; direct user control is limited.
& Always fits the layout exactly within the specified rectangular boundary. \\

\midrule
Topological guarantees
& No formal guarantees on planarity, biconnectivity, or removal of separating triangles.
& Focuses on visual realism; rectangular dual realizability is not ensured.
& Explicitly constructs a separating-triangle-free plane triangulation, ensuring the existence of a valid rectangular dual consistent with the constraints. \\

\midrule
Role of data
& Requires large annotated datasets for training; performance depends on training distribution.
& Strongly data-driven; generalization beyond trained building types may be limited.
& Fully algorithmic and independent of training data; applicable without example floorplans. \\

\midrule
Nature of output
& Raster floorplans with extracted room boxes; dimensions reflect learned data patterns.
& Vector layouts resembling professional designs; sizes and proportions follow learned distributions.
& Families of dimensionless rectangular floor plans that can later be dimensioned using optimization or rule-based methods. \\

\midrule
User control and interpretability
& User specifies graph and boundary, but geometric realization is learned and less transparent.
& Designer sketches relationships; internal generative decisions are not directly interpretable.
& Each modification (connectivity, biconnectivity, triangulation, separating-triangle handling) corresponds to a clear algorithmic step, making the process transparent and verifiable. \\

\bottomrule
\end{tabular}
\end{table}

\section{Our Proposed Work}
\label{Our}
We propose an interactive prototype that transforms an input door-connectivity specification into a graph suitable for floor-plan construction. 
From this backbone, the system generates floor plans that can be selected and further refined. Unlike approaches that assume the input constraint graph is fixed in advance, our method treats the evolving graph as the main design object. 
Each modification step is explicit, providing clear insight into how feasibility is established.

\subsection{Inputs}
The user specifies (i) a minimum door-adjacency graph representing only essential connections, (ii) a non-adjacency set representing forbidden connections, and (iii) a rectangular boundary template.
This input format aligns with early-stage briefs while remaining precise enough to support correctness by construction processing.

\subsection{Certified backbone construction}
Our developed prototype implements a staged augmentation pipeline, where each stage has a focused structural goal and can be inspected and validated:
\begin{itemize}
  \item \textbf{Connectivity augmentation:} if the input graph contains multiple connected components(blocks), the system adds the minimum number of edges needed to make it connected, rejecting any candidate edge that violates forbidden/non-adjacency edges.
  \item \textbf{Bi-connectivity augmentation:} the connected graph is augmented to eliminate articulation points. This ensures that the graph remains connected even after removing a single vertex, while still respecting the given non-adjacency constraints.
  \item \textbf{Constrained planar embedding and triangulation:} The bi-connected graph is then embedded in the plane and further refined by adding diagonals that satisfy the given constraints. This process produces a triangulated structure suitable for the subsequent synthesis stage.
  \item \textbf{Handling separating triangles (User mode-specific):} When the construction requires strictly rectangular modules, separating triangles are handled by carefully removing and adding edges so that a rectangular floor plan is possible. In contrast, if more flexible orthogonal floor plans/layouts are allowed, separating triangles are resolved by introducing additional nodes into the graph.
\end{itemize}
\subsection{Interaction and Downstream Synthesis}

After the combinatorial backbone has been formally constructed and verified, i.e., PTPG or PTG, the prototype produces multiple dimensionless floor plan/layout alternatives and enables interactive refinement. The system visually distinguishes between edges defined by the user and those introduced during algorithmic augmentation. Whenever the user modifies adjacency requirements or non-adjacency/forbidden-contact edge pairs, the system immediately checks feasibility and reports whether the updated specification remains valid.\\
This framework addresses a practical limitation in existing approaches. Data-driven methods often generate visually convincing floor plans, but the underlying structural decisions that inform them are not directly transparent. Conversely, traditional graph-based methods provide formal correctness guarantees but typically operate in a fixed, non-interactive workflow. 
In contrast, the proposed prototype keeps the combinatorial structure editable, transparent, and formally validated before any geometric optimization or rendering is performed.
\section{Methodology}
\label{Meth}
This section presents the graph-theoretic framework underlying our approach to architectural floor plan generation, focusing on the algorithmic workflow and practical considerations. The method starts with user-provided inputs: a door-based adjacency graph and explicit non-adjacency constraints, which together describe essential spatial relationships. To obtain a connected, structurally valid representation suitable for floor-plan synthesis, the graph is augmented with additional edges corresponding to interior wall adjacencies. A key aspect of this stage is the treatment of separating triangles, as their presence directly affects the geometric realizability of the resulting floor plan. When the objective is to generate strictly rectangular floor plans, the augmented graph must remain free of separating triangles while respecting all input constraints, since such configurations ensure compatibility with rectangular modules. In contrast, when non-rectangular room shapes are allowed, separating triangles are handled by introducing additional vertices that correspond to new spatial modules, enabling the synthesis of orthogonal floor plans.\\
The proposed framework is structured as a five-stage pipeline, where each stage performs a specific structural or geometric refinement. Only the stages necessary for the user’s design objectives are executed. The process begins with the user-provided minimal constraint graph. The first step is to check whether the graph is connected. If it is disconnected, Algorithm~\ref{OF-CON} adds the minimum required edges to obtain a connected graph while strictly respecting all specified non-adjacency constraints. If the graph is already connected, the method proceeds to the next stage. In the second stage, Algorithm~\ref{BY} augments the connected graph to make it biconnected. This ensures that the graph remains connected after the removal of any single vertex and that all required adjacency relations are preserved during this augmentation. The resulting biconnected graph is then converted into the intermediate representation needed for further processing. Subsequently, triangular configurations are treated according to the desired module geometry. When strictly rectangular modules are required, separating triangles are removed in a constraint-preserving manner so that all adjacency conditions remain satisfied.\\
After these steps, the output is either a plane triangulated graph (PTG) or a properly triangulated plane graph (PTPG), depending on the chosen layout model. This certified structure forms the input to the floorplan synthesis stage. The details of each stage are described in the following section
\subsection{Minimal Connectivity to 1-connected Graph.}\label{M-C}
\begin{algorithm}
\caption{: $Connectivity$ $with$ $Non$-$adjacency$}
\label{OF-CON}
\begin{algorithmic}[1]
\Function{$Connectivity$}{$G(V,E)$, $non{\_}adj{\_}set$}
  \State $augmentation{\_}edges \gets \{\}$.
  \State $components \gets Connected{\_}Components(G)$, $m$ $=$ $|{components}|$. 
  \State $outerfaces \gets \{\}$. \Comment{For each component $i$, store vertices lying on its outer face boundary}
  \If{$length(components)$ $==$ $1$}
  \State  \Return \textcolor{red}{$G_{C}(V_C,E_C)$ = $G(V,E)$}.
  \Else
  
  \For{$i \gets 0$ \textbf{to} $length(components)-1$}  \Comment{Compute outerface boundary of each component via planar embedding}  
     \State $outerfaces[i] \gets Extract\_Outerface(components[i])$.
  \EndFor
  \State $possible{\_}edges \gets Generate\_Possible\_Edges(outerfaces, non{\_}adj{\_}set)$.
  \State Initialize $root[i] \gets i$ for $i = 0.....|components|-1$.
 \For{$(u,v) \in possible{\_}edges$} \Comment{For each vertex $u$, $component\_ID(u)$ gives the ID of the component containing $u$.} 
     \State $c_u \gets Component\_ID(u)$, $c_v \gets Component\_ID(v)$.
     \If{$Find\_Comp(c_u) \neq Find\_Comp(c_v)$}
        \State $Union\_Comp(c_u,c_v)$.
        \State $augmentation{\_}edges \gets augmentation{\_}edges \cup \{(u,v)\}$, $m\gets m - 1$.
     \EndIf
  \EndFor

  \State $E \gets E \cup augmentation{\_}edges$.
  \State \Return \textcolor{red}{$G_{C}(V_C,E_C)$ = $G(V,E)$, $m$}.
  \EndIf
 \Function{$Extract\_Outerface$}{$component$}
  \State \Return $\{\, v \in V_C \;\mid\; \text{$v$ lies on the boundary of $component$} \,\}$.

\EndFunction 
\Function{$Generate\_Possible\_Edges$}{$outerfaces$, $non{\_}adj{\_}set$}
  \State $possible{\_}edges \gets \{\}$.
  \For{$i \gets 0$ \textbf{to} $length(outerfaces)-1$}
    \For{$j \gets i+1$ \textbf{to} $length(outerfaces)-1$}
      \For{$u \in outerfaces[i]$}
        \For{$v \in outerfaces[j]$}
          \If{$(u,v) \notin non{\_}adj{\_}set$}
            \State $possible{\_}edges \gets possible{\_}edges \cup \{(u,v)\}$.
          \EndIf
        \EndFor
      \EndFor
    \EndFor
  \EndFor
  \State \Return $possible{\_}edges$.
\EndFunction
\Function{$Find\_Comp$}{$u$}
    \If{$root[u] \neq u$}
       \State $root[u] \gets Find\_Comp(root[u])$.
    \EndIf
    \State \Return $root[u]$.
  \EndFunction
\Function{$Union\_Comp$}{$u,v$}
    \State $root_u \gets Find\_Comp(u), root_v \gets Find\_Comp(v)$.
    \If{$root_u \neq root_v$}
       \State $root[root_v] \gets root_u$.
    \EndIf
  \EndFunction
  \EndFunction
  \If{$m>1$} \Comment{(fallback): ignore $non$-$adj$-$constraints$ and retry replacements}
   \State $G_{C}(V_C,E_C) \gets Connectivity$($G_C(V_C,E_C)$, $non{\_}adj{\_}set$ = $\phi$). 
  \EndIf

\end{algorithmic}
\end{algorithm}
Connectivity is a basic requirement for floor-plan synthesis. If the constraint graph is disconnected, the rooms form separate components that cannot be realized as a single layout. Therefore, the graph must be at least connected (1-connected) so that all spaces belong to one coherent system. A connected graph provides the foundation for later steps such as biconnectivity and triangulation, which depend on a unified structure to produce a consistent geometric layout. Earlier work, including Roth et al.~\cite{roth1982turning}, also notes that valid floor plans are derived from connected and subsequently triangulated layout graphs. When users specify only door-based relations, the resulting graph may not be connected. In such cases, we introduce additional wall-adjacency edges to ensure connectivity before further processing. This augmentation is performed carefully so that no given non-adjacency constraints are violated. Algorithm~\ref{OF-CON} formalizes this step by inserting edges only when necessary.\\\\
\textbf{Illustration of Algorithm \ref{OF-CON} (see Figure \ref{F-1}):}\\\\
\textbf{1. Steps 1–6}: begin by taking the input graph $G(V, E)$ together with the user-specified non-adjacency constraints shown in  Figure \ref{F-1}(a). We then determine how the graph is partitioned by identifying all of its connected components. When this check reveals that the graph is already fully connected, meaning only a single component is present, the process stops, and the original graph is returned. This early termination prevents unnecessary augmentation and ensures that additional processing occurs only when the input actually contains multiple disconnected parts.\\\\
\textbf{2. Steps 7–11:} are invoked when the input graph is divided into more than one connected component. In this stage, each component is examined individually, beginning with the extraction of the vertices that appear on its outer boundary, shown in Figure \ref{F-1}(b). For the example ($G(V,E)$) under consideration, these boundary sets are $\{2, 6, 8, 7, 0, 13, 1\}$, $\{9, 11, 12, 10\}$, and $\{3, 4, 5\}$. Using these outerface vertices, the algorithm then forms all possible cross-component pairs and removes any pair that conflicts with the specified non-adjacency constraints. The pairs that remain after this filtering step represent the viable edges through which disconnected components may be linked. For the graph $G(V, E)$ in the example, this procedure yields 51 such admissible connections/edges while satisfying non-adjacency constraints (see Figure \ref{F-1}c):\\ \textbf{[21(outerface[1]--outerface[2]) + 18(outerface[1]--outerface[3]) + 12(outerface[2]-outerface[3])] = 51}, which are stored in the $possible\_edges$ set.\\\\
\textbf{3.  Steps 12–16:} instantiate a union–find data structure in which each connected component is initially represented as its own set (for example, all vertices of component[1] = $\{0,2,7,8,6,1,13,14\}$ are assigned the same component index; see Figure \ref{F-1}(b)). The algorithm then processes the acceptable/candidate edges from $possible\_edge\_set$ sequentially: for each edge $(u,v)$, it checks whether $u$ and $v$ belong to distinct union–find sets; if so, it accepts $ (u, v) $ as an augmentation edge and merges the two sets/components. For instance, $(2,9)$ is accepted because vertex 2 belongs to component[1] and 9 to component[2], and the pair does not violate non-adjacency constraints; after this union, a subsequent candidate edge such as $(2,11)$ is skipped because 11 now belongs to the same merged set as 2. This accept-and-merge process continues for the remaining candidate edges until all components are combined into a single connected graph. Only those edges that connect previously disjoint components are stored in the $augmentation\_edges\_set$.\\\\
\textbf{4. Steps 17–18:} conclude Algorithm \ref{OF-CON} by incorporating the selected augmentation edges into the existing edge set $E$, thereby completing the connectivity of the graph. Once these edges have been inserted, the algorithm produces the final connected graph $G_C$, shown in Figures \ref{F-1}(c) and \ref{F-1}(d), which reflects the fully unified structure resulting from the earlier augmentation steps.\\\\
\textbf{5. Steps 19–37:} of Algorithm \ref{OF-CON} rely on several helper/additional functions that perform the supporting work necessary for the main method to run smoothly. The $Extract\_Outerface$ function generates the vertices on the boundary cycle of each component in the input plane graph, as illustrated in Figure \ref{F-1}(b) (see outerface[1], outerface[2], outerface[3]). The $Generate\_Possible\_Edges$ function then checks all cross-component pairs drawn from these boundary vertices and removes any pair that violates the non-adjacency constraints. The $union$ function keeps track of the component to which each vertex belongs, using the operations $Find_Comp$ and $Union_Comp$; for example, once an edge in the $augmentation\_edges$ set merges components 1 and 2, any later edge whose endpoints lie in these two components is skipped, since they have already been joined. Every edge in the $augmentation\_edges$ set satisfies the non-adjacency constraints, and each chosen edge genuinely reduces the number of disconnected components. Through these checks, Algorithm \ref{OF-CON} ultimately produces a valid connected graph $G_C$, as shown in Figure \ref{F-1}(a–d).\\\\
\textbf{5. Steps 38–39:} If the current constraints do not allow additional edges to be inserted to achieve full connectivity, a fallback step is applied. This step temporarily ignores non-adjacency constraints and re-executes the connectivity function to merge all components.\\\\
By applying Algorithm \ref{OF-CON}, we obtain a connected graph $G_C$ derived from the user-specified input graph $G$ together with the stated non-adjacency constraints (see Figure \ref{F-1}(a–d)). 
Since we are treating the input graph as a plane drawing, the outer face of each connected component is well-defined. In Algorithm~\ref{OF-CON}, the vertices incident to the outer face are identified prior to any augmentation, and all additional edges are inserted only afterward; consequently, the extraction of outer-face vertices is independent of the augmentation step and does not require recomputation or planarity verification. Once the final connected graph $G_C$ is obtained, a plane embedding of $G_C$ serves as input to the subsequent algorithm, which further processes it to achieve bi-connectivity.
\begin{figure}
\centering
    \includegraphics[width = 0.95 \textwidth]{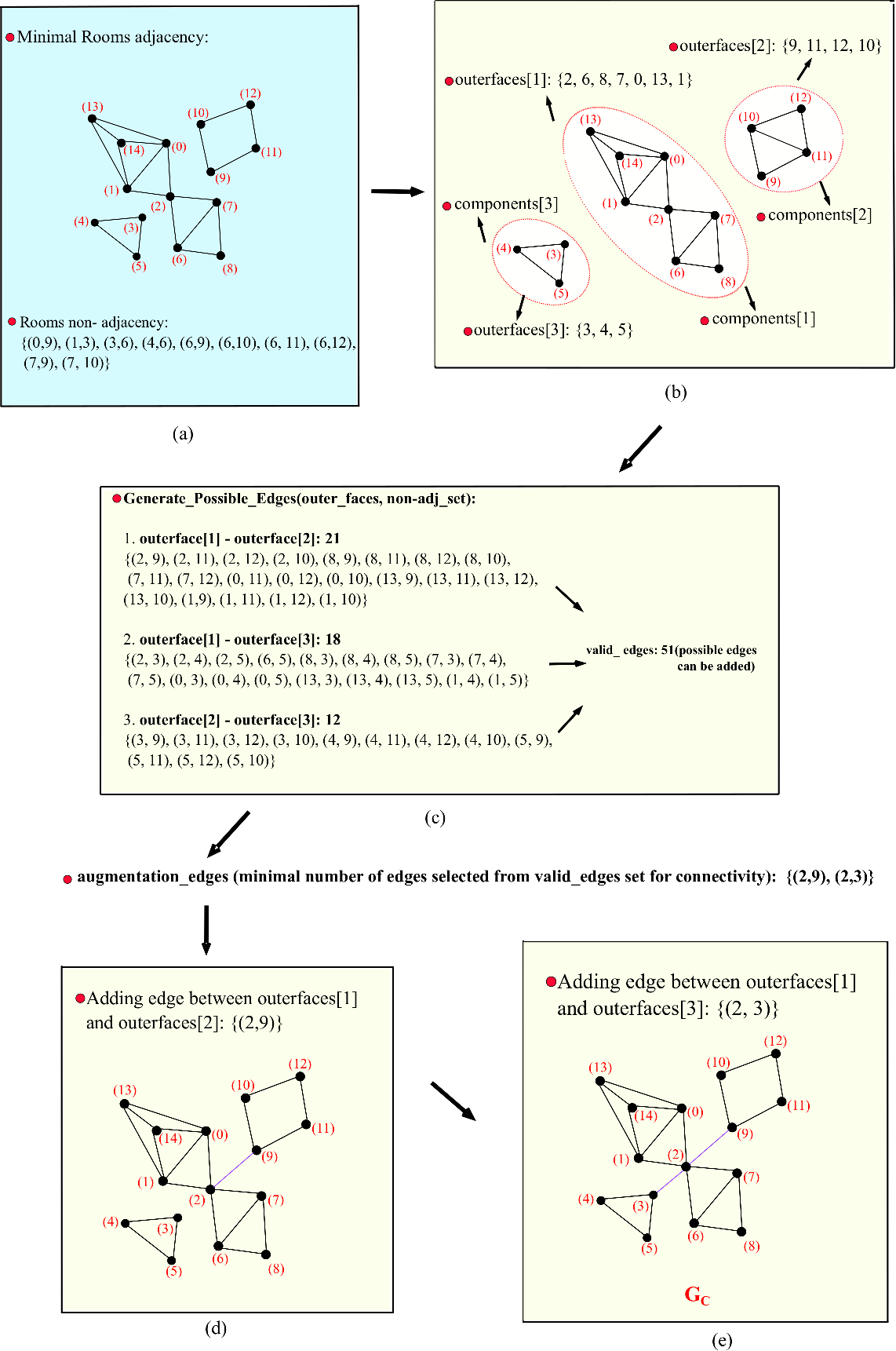}
    \caption{(a-e) Construction of a connected graph $G_C$ while preserving the specified input constraints.}
    \label{F-1}
\end{figure}
\subsection{Connectivity graph to Bi-connected graph}\label{C-BY}
\begin{algorithm}
\caption{: $Biconnectivity$ $with$ $non-adjacency$ }
\label{BY}
\algnewcommand\To{\textbf{to }}
  \begin{algorithmic}[1]
  \Function{$Biconnectivity$}{$G_C(V_C,E_C)$,$articulation{\_}points$, $non{\_}adj{\_}set$}
  \State $selected{\_}edges$ $=$ \{\}.
   \For{$v \in articulation{\_}points$}
  \State $blocks \gets Blocks$($G_C$,$v$).
  \State $valid{\_}edges \gets Valid{\_}Edges$($blocks$, $non{\_}adj{\_}set$).
  \State $bcn{\_}edges \gets Connect{\_}Blocks$($blocks$,$valid{\_}edges$).
  \State $selected{\_}edges \gets selected{\_}edges$ $\cup$ $bcn{\_}edges$.
  \EndFor
        \State $E \gets E$ $\cup$ $selected{\_}edges$.
  \State \Return \textcolor{red}{$G_{B}(V_B,E_B)$ $=$ $G_C(V_C,E_C)$}.

  \Function{$Blocks$}{$G$,$r$} 
   \State $G'(V',E')$ $=$ $G(V,E)$.
   \State $\{u_1,u_2,...., u_k\} \gets nbd(v) $, $V' \gets V'$ $-$ $\{v\}$.  \Comment{nbd($v$) = vertices adjacent to $v$ in $G$}
   \For{$i \gets 1$ \To $k$}
   \State $E' \gets E'$ $-$ $\{(v,u_i)\}$.
   \EndFor
   \State $Components \gets Connected{\_}components(G')$. \Comment{$Connected{\_}components(G')$: Returns $component$: generator of sets.}
   \For{$i \gets 0$ \textbf{to} $length(Components)-1$}  \Comment{Compute outerface boundary of each component via planar embedding}  
     \State $blocks[i] \gets Extract\_Outerface(Components[i])$.
  \EndFor
    \State \Return $blocks$.  \Comment{$blocks$ = \{$ B_1,B_2,...., B_k$\}.}
    \EndFunction
    \Function{$Valid{\_}Edges$}{$blocks$, $non{\_}adj{\_}set$} 
   \State $valid{\_}edges$ = \{\}.
    \For{$i \gets 0$ \To $length(blocks)-1$}
    \For{$j \gets i + 1$ \To $length(blocks)-1$}
    \State $BL_1\gets blocks[i]$, $BL_2 \gets blocks[j]$ 
    \For{ $u \in BL_1$}
    \For{ $v \in BL_2$}
    \If{$(u,v) \not\in non{\_}adj{\_}set$}
    \State $valid{\_}edges \gets valid{\_}edges$ $\cup$ $\{(u,v)\}$.
    \EndIf
    \EndFor
       \EndFor
    \EndFor
       \EndFor
       \State \Return $valid{\_}edges$. 
    \EndFunction
    \Function{$Connect{\_}Blocks$}{$blocks$,$valid{\_}edges$, $non{\_}adj{\_}set$} 
  \State $block{\_}graph \gets Empty{\_}Map()$.\Comment{Maps a block index to its corresponding set of adjacent blocks.}
  \State $b \gets length(blocks)$.
  \State $node{\_}block \gets Empty{\_}Map()$.\Comment{Associates each node with the block it belongs to.} 
  \For{$i \gets 0$ \To $b-1$}
  \For{$u \in blocks[b]$}
  \State $node{\_}block[u] \gets b$.
  \EndFor
  \EndFor
  \State $parent \gets Array[0,.....,b-1]$.
   \For{$i \gets 0$ \To $b-1$}
   \State $parent[i] \gets i$. \Comment{Every block initially serves as its own representative.}
   \EndFor
   \Function{$Find[u]$}{} \Comment{To find the root of a block in connected components.}
   \If{$parent[u] \neq u$}
    \State $parent[u] = Find[parent[u]]$.
   \EndIf
   \State \Return $parent[u]$.
   \EndFunction
    \Function{$Union$}{$u,v$} \Comment{Union function to connect two blocks.}
    \State $root_u = Find[u]$, $root{\_}v = Find[v]$.
    \If{$root{\_}u \neq root{\_}v$}
     \State $parent[root{\_}v] = root{\_}u$.
    \EndIf
   \EndFunction
   \State $valid{\_}edges.sort()$.  \Comment{ Sort valid edges by weight (default: unit weight).}
   \State $selected{\_}edges$ $=$ \{\}.
   \For{$u, v  \in valid{\_}edges$}
   \State $block{\_}u = node{\_}block[u]$, $block{\_}v = node{\_}block[v]$.
   \If{$Find[block{\_}u] \neq Find[block{\_}u]$}
   \State $Union (block{\_}u, block{\_}v)$. 
   \State $selected{\_}edges \gets selected{\_}edges$ $\cup$ $\{(u,v)\}$.
   \State $block{\_}graph[block{\_}u] \gets block{\_}graph[block{\_}u]$ $\cup$ $\{block{\_}v\}$.
   \State $block{\_}graph[block{\_}v] \gets block{\_}graph[block{\_}v]$ $\cup$ $\{block{\_}u\}$.
   \EndIf
   \EndFor
\State $connected\_blocks \gets \left|\{\, Find[u] \mid u \in [0, b-1] \}\right|$. \Comment{Check if all blocks are connected.} 
 \If{$connected\_blocks > 1$} \Comment{fallback: Add non-adjacent edges if necessary to connect all blocks.}
 \For{$u,v \in non{\_}adj{\_}set$ $and$ $u,v \in blocks$} 
 \State $block{\_}u = node{\_}block[u]$, $block{\_}v = node{\_}block[v]$.
  \If{$Find[block{\_}u] \neq Find[block{\_}u]$}
   \State $Union (block{\_}u, block{\_}v)$, $selected{\_}edges \gets selected{\_}edges$ $\cup$ $\{(u,v)\}$.
   \State $block{\_}graph[block{\_}u] \gets block{\_}graph[block{\_}u]$ $\cup$ $\{block{\_}v\}$, $block{\_}graph[block{\_}v] \gets block{\_}graph[block{\_}v]$ $\cup$ $\{block{\_}u\}$.
   \EndIf
 \EndFor
\EndIf
\State \Return $selected{\_}edges$.
    \EndFunction
 \EndFunction
\end{algorithmic} 
\end{algorithm}
Koźmiński et al. \cite{kozminski1985rectangular} showed that generating a rectangular or orthogonal floor plan requires the underlying layout graph to be both bi-connected and triangulated. Consequently, when deriving a rectangular floor plan from a user-specified minimal connectivity graph, the first step is to augment the previously obtained one-connected graph (from Algorithm \ref{OF-CON}) so that it becomes bi-connected, after which triangulation can be performed. Our proposed algorithm (Algorithm \ref{BY}) provides a procedure for constructing a bi-connected graph while respecting user-defined constraints. Its objective is to introduce additional adjacencies only where necessary, ensuring that the graph remains connected even after the removal of any single vertex. Our proposed Algorithm \ref{BY} takes as input a plane-connected graph $G_C$ and articulation points in $G_C$ (using the approach proposed by Tarjan \cite{tarjan1972depth}), identifies the blocks associated with them, and then selectively adds edges within these blocks. Throughout this process, it adheres strictly to the supplied Non-Adjacency list, ensuring that no non-adjacency pairs are ever connected. To demonstrate the algorithm's operation, we illustrate Algorithm \ref{BY} step-by-step using a representative example with minimal connectivity.\\\\
\textbf{Illustration of Algorithm \ref{BY} (see Figure \ref{F-2}):}\\\\
\textbf{1. Steps 1--7:} We begin with the plane-connected graph $G_C(V_C, E_C)$ obtained from Algorithm~\ref{OF-CON}, together with the list of articulation points and the set of user-specified non-adjacency constraints (see Figure~\ref{F-2}(a)). The idea is to add a minimal extra edges so that the graph becomes bi-connected, while respecting the input constraints as far as possible. An empty set $selected\_edges$ is created to store all final $bcn\_edges$. The algorithm then scans the list of articulation points one by one. For a fixed articulation point $v$, the procedure \texttt{Blocks$(G_C,v)$} is called to see how the neighbourhood of $v$ breaks into connected components when $v$ is removed, and \texttt{Valid\_Edges} lists all permitted candidate edges between those component vertices while satisfying non-adjacency constraints. Finally, \texttt{Connect\_Blocks} chooses a suitable subset of these candidate edges (i.e., the minimal number of edges added to make the graph bi-connected), and the chosen edges are added to $selected\_edges$.\\\\
\textbf{2. Steps 8--9:} After all articulation points have been processed, the set $selected\_edges$ contains every new edge that we intend to insert. These are now added to the original edge set $E_C$ to form $E' = E_C \cup selected\_edges$, and the resulting graph is denoted by $G_B(V_B, E_B)$. In the running example, the red edges appearing in Figure~\ref{F-2}(e) are exactly those stored in $selected\_edges$. The graph $G_B$ is the bi-connected graph that will be passed to the triangular phase construction step.\\\\
\textbf{3. Steps 10--18:} The helper function \texttt{Blocks$(G_C,v)$} determines how the graph around $v$ is organized. Conceptually, we delete $v$ and all edges incident to it from $G_C$, and then compute the connected components of the remaining graph $G'$. For each such component, we will compute its outerface vertices and record them as a block, so we obtain a collection
\[
blocks = \{B_1,B_2,\ldots,B_\ell\}.
\]
These blocks represent the regions that would separate if $v$ were removed. For the example in Figure~\ref{F-2}, when $v=9$ we obtain two blocks
\[
blocks[1] = \{10,11,12\}, \qquad
blocks[2] = \{0,13,1,2,6,8,7,3,4,5\},
\]
as shown in Figure~\ref{F-2}(b). The block decompositions for $v=3$ and $v=2$ are displayed in Figures~\ref{F-2}(c) and~\ref{F-2}(d).\\\\
\textbf{4. Steps 19--28:} Once the blocks are known, the function \texttt{Valid\_Edges(blocks, non\_adj\_set)} builds a pool of valid edges that do not break the non-adjacency rules. For every unordered pair of distinct blocks $B_i$ and $B_j$, it looks at all vertex pairs $(u,w)$ with $u \in B_i$ and $w \in B_j$. If $(u,w)$ appears in the non-adjacency set, that pair is discarded; otherwise, it is inserted into the set $valid\_edges$. In this way, $valid\_edges$ collects every allowed way of joining different blocks around the same articulation point. The edges that are finally drawn in Figures~\ref{F-2}(b)--(d) (such as $(0,10)$ for vertex $9$ or $(4,2)$ for vertex $3$) come from this pool.\\\\
\textbf{5. Steps 29--63:} The primary task of determining which edges should be retained is carried out by the function \texttt{$Connect\_Blocks$}. Initially, each block is placed in its own set in a union-find data structure, and for each vertex, we store the index of the block it belongs to. The edges in $valid\_edges$ are then inspected one by one (since all edges have equal weight in our setting, any fixed ordering is sufficient). For each candidate edge $(u,w)$, we identify the blocks $block\_u$ and $block\_w$ that contain $u$ and $w$. If these blocks currently reside in distinct union find sets, the edge $(u,w)$ is accepted: it is appended to $selected\_edges$, the two sets are merged via a \texttt{Union} operation, and the auxiliary block graph is updated accordingly. This accept-then-merge process continues until no remaining candidate edge can reduce the number of disconnected block sets. If multiple block sets are still disconnected at this point, the fallback mechanism of our proposed algorithm makes another pass over all admissible pairs to connect the remaining components of the block graph, while ignoring the non-adjacency constraints.\\
In the Example \ref{F-2}, for articulation point $9$ there are $27$ $valid_edges$ set (It is generated in the same manner that we show during identifying $valid_edges$ in Algorithm \ref{OF-CON}: see Figure \ref{F-1}c). The edge $(0,10)$ unifies as blocks $blocks[1]$ and $blocks[2]$ (Figure~\ref{F-2}(b)). For instance, $(10,0)$ (belongs to $valid_edges$) is admitted into $bcn\_edges$ because vertex $10$ belongs to blocks[1] and vertex $0$ belongs to blocks[2], and this pairing does not violate any non-adjacency rules defined by the user. After this merge, the next candidate edge, such as $(10,13)$ (which belongs to $valid\_edges$), is rejected because $10$ and $13$ now lie in the same merged set. This merging process continues across the remaining candidate edges in the $valid\_edges$ set until all components are combined into a single connected graph. Only the edges that connect different components are added to $bcn\_edges$. \\
Similarly, for articulation point $3$, the edge $(4,2)$ performs the merging of its two adjacent blocks (see Figure~\ref{F-2}(c)). A similar process for articulation point $2$ yields a small collection of edges such as $(6,1)$, $(6,5)$, and $(7,11)$ that ensure all blocks incident to vertex $2$ ultimately merge into one merged block (Figure~\ref{F-2}(d)). Each edge stored in $selected\_edges$ therefore satisfies the non-adjacency rules, joins vertices from different blocks of some articulation point, and decreases the number of disconnected block sets.
Combining all steps, Algorithm~\ref{BY} converts the connected graph $G_C$ generated by Algorithm~\ref{OF-CON} into the bi-connected graph $G_B$ shown in Figure~\ref{F-2}(e). The resulting graph $G_B$ contains no articulation points and thus serves as the robust starting structure for the subsequent triangulation algorithm.\\\\
By applying Algorithm \ref{BY}, we obtain a bi-connected graph $G_B$ derived from the input graph $G_C$ together with the stated non-adjacency constraints (see Figure \ref{F-2}(a–e)). 
Since we are treating the input graph as a plane drawing, the outer face of each connected component is well-defined. In Algorithm~\ref{BY}, the vertices incident to the outer face are identified prior to any augmentation, and all additional edges are inserted only afterward; consequently, the extraction of outer-face vertices is independent of the augmentation step and does not require recomputation or planarity verification. Once the final bi-connected graph $G_B$ is obtained, a plane embedding of $G_B$ serves as the input for the subsequent algorithm, which further processes it to achieve triangulation.
\begin{figure}
\centering
    \includegraphics[width = 0.97 \textwidth]{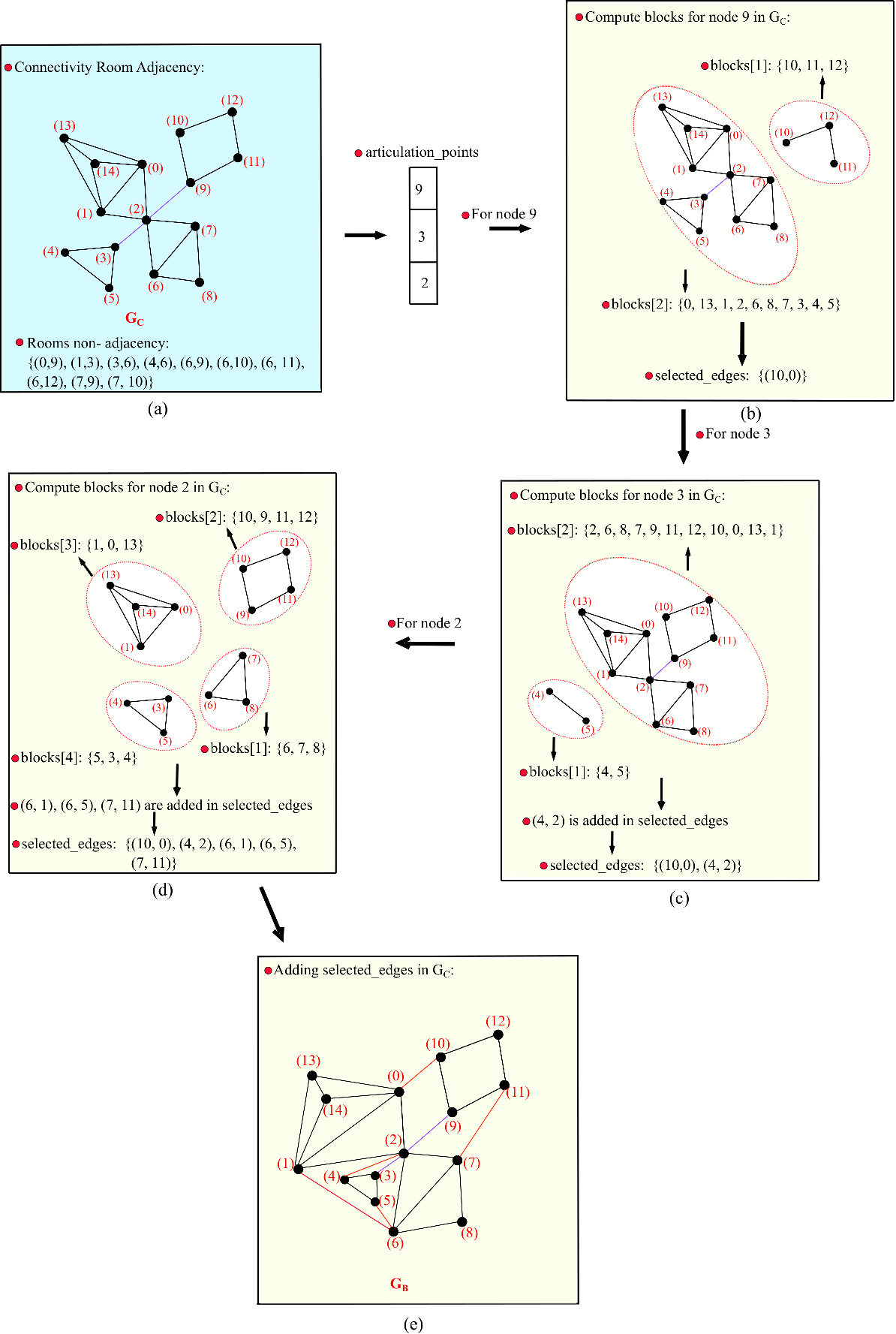}
    \caption{(a-e) Construction of a Bi-connectivity graph $G_B$ while preserving the specified input constraints.}
    \label{F-2}
\end{figure}
\subsection{Bi-connected to Triangulation}\label{BY-TRI}
\begin{algorithm}
\caption{: $Triangulation$ $with$ $non$-$adjacency$}
\label{TRI}
\algnewcommand\To{\textbf{to }}
\begin{algorithmic}[1]

\Function{$Triangulation$}{$G_B(V_B, E_B), non\_adj\_set$}

\State $faces \gets Find\_Faces(G_B)$.
\State $non\_tri\_faces$ $\gets []$.
\For{each $face \in faces$}\Comment{$Vert(face)$: Returns the list of vertices forming the face.}
    \State $vertices \gets Vert(face)$. \Comment{$Is\-\_Exterior\_Face$: Checks if the given vertices form the outer boundary of the graph.}
    \If{$vertices > 3$ and $Is\_Exterior\_Face(vertices, G_B) = false$}
\State $non\_tri\_faces.append(vertices)$.
    \EndIf
  \EndFor
\State $Triangulate\_edges \gets []$.
\State $Triangulate\_edges \gets Clip\_Triangulation(non\_tri\_faces, G_B, non\_adj\_set)$.
 \State $E \gets E \cup Triangulate{\_}edges$.
   \State \Return \textcolor{red}{$G_{T}(V_T,E_T)$ $=$ $G_{B}(V_B,E_B)$}.
\Function{$Find\_Faces$}{$G_B(V,E)$}
  \State $faces \gets []$, $E'$ $=$ $E$.
    \While{$E' \neq \varnothing$}
    \State Choose $s \in E'$.
    \State $P \gets [s]$.
    \State $c \gets s$.
    \While{true}  \Comment{$Get\_Next\_Edge$: Given $e = (u,v)$, returns the next edge $(v,w)$ around $v$ in $G_B$.}
      \State $n \gets \text{$Get\_Next\_Edge$}(c, G_B)$.
      \If{$n = s$}
        \State $\text{faces.append}(P)$.
        \State \textbf{break}
      \Else
        \State $P.\text{append}(n)$.
        \State $c \gets n$.
        \State $E'.\text{remove}(n)$.
      \EndIf
    \EndWhile
  \EndWhile
  \State \Return $faces$.
\EndFunction
\Function{$Clip\_Triangulation$}{$non\_tri\_faces, G_B, non\_adj\_set$}
  \State $new\_edges \gets []$.
  \For{each $face \in non\_tri\_faces$}
    \State $vertices \gets Vert(face)$. \Comment{$Vert(face)$: Returns the ordered list of vertices forming the face.}
    \While{$|vertices| > 3$}
      \State $clip\_found \gets false$. \Comment{$Clip\_found$ is a boolean variable}
      \For{$i \gets 0$ \To $|vertices|-1$}
        \State $(a,b,c) \gets Consecutive(vertices,i)$.
        \If{$Is\_Valid\_Diagonal$($a, c, vertices, face, non\_adj\_set$)}
           \State $new\_edges \gets new\_edges.append((a,c))$.
         \State $vertices \gets vertices.remove(b)$.
          \State $clip\_found \gets true$.
          \State \textbf{break}.
        \EndIf
      \EndFor
      \If{$clip\_found == false$} \Comment{fallback: To connect all blocks ignore non-adjacency constraints.}
        \State $(a,b,c) \gets Consecutive(vertices,0)$.
        \State $new\_edges \gets new\_edges.append((a,c))$.
         \State $vertices \gets vertices.remove(b)$.
      \EndIf
    \EndWhile
  \EndFor
  \State \Return $new\_edges$.
\EndFunction
 \Function{$Consecutive$}{$vertices,i$} \Comment{$vertices = (v_1, v_2, \dots, v_l)$ is an ordered list of vertices}
  \State $a \gets v_i$.
  \State $b \gets v_{(i+1) \bmod l}$.
  \State $c \gets v_{(i+2) \bmod l}$.
  \State \Return $(a,b,c)$.
\EndFunction

\Function{$Is\_Valid\_Diagonal$}{$a, c, vertices, face, non\_adj\_set$}

  \If{$(a,c) \in non\_adj\_set$}
    \State \Return false.
  \EndIf
 
  \For{each edge $(v_i, v_{i+1}) \subseteq vertices$}  \Comment{Check that diagonal does not intersect any $face$ edges}
    \If{$(a,c)$ intersects $(v_i, v_{i+1})$ \textbf{and} $(v_i,v_{i+1}) \neq (a,b)$, $(b,c)$}
      \State \Return false.
    \EndIf
  \EndFor
  \If{$(a,c) \in \mathrm{Diag}(vertices)$ \textbf{ and } $(a,c) \subseteq \mathrm{int}(vertices)$}
 \Comment{Check that diagonal is inside the polygon face.}
    \State \Return true. \Comment{$\mathrm{Diag}(vertices)$ denotes the set of all possible diagonals of polygon $face$.}
  \Else \Comment{$\mathrm{int}(vertices)$ denotes the interior region of polygon $face$.}
    \State \Return false
  \EndIf
\EndFunction
\EndFunction
\end{algorithmic}
\end{algorithm}
 \begin{figure}
\centering
    \includegraphics[width = .5 \textwidth]{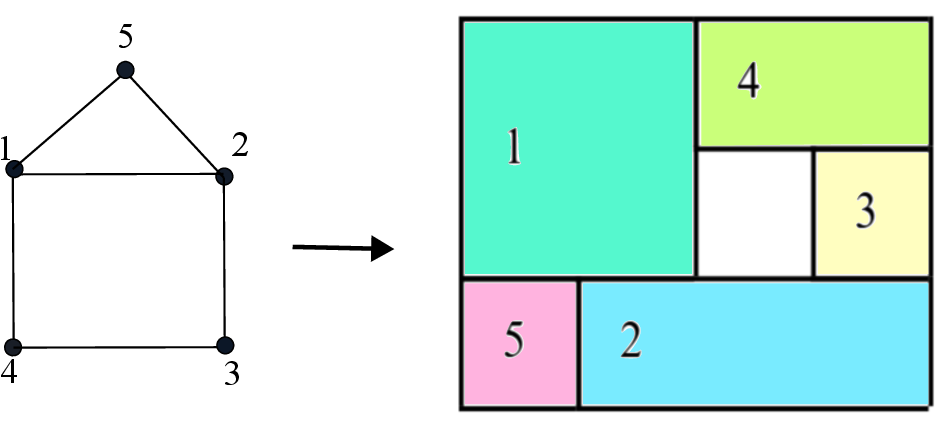}
    \caption{Additional space is created in the floor plan corresponding to the graph with a non-triangular face (1,2,3,4).}
    \label{EX-1}
\end{figure}
After constructing a bi-connected graph from the user-specified door connectivity, the next stage performs triangulation to introduce interior wall adjacencies, eliminating unused space within the rectangular floor plan. Failure to triangulate, i.e., leaving any face non-triangular, can cause the floor plan/layout generation procedure to create voids or inefficiencies (see Figure \ref{EX-1}). To address this, we proposed Algorithm \ref{TRI} that triangulates a bi-connected graph while respecting user-defined non-adjacency constraints. The method compiles all non-triangular faces and non-adjacent edges, then iteratively decomposes each face into triangles. Using an ear-clipping strategy inspired by Mei et al. \cite{mei2012ear}, it attempts to insert diagonals that neither violate user-imposed adjacency restrictions nor introduce disallowed/non-adjacency edges. When such valid diagonals exist, they are added to the graph, incrementally refining spatial relationships as illustrated in Figure \ref{F-3}, where six interior adjacencies are introduced. If a face admits no diagonal other than those explicitly prohibited non-adjacency edges, the algorithm temporarily relaxes these restrictions to ensure the graph becomes fully triangulated. This fallback mechanism guarantees a complete triangulation while preserving user constraints whenever feasible. To demonstrate the algorithm's operation, we illustrate Algorithm \ref{TRI} step-by-step using a representative Bi-connectivity graph as input with respect to an example \ref{F-3}.\\\\
\textbf{Illustration of Algorithm \ref{TRI} (see Figure \ref{F-3}):}\\\\
\textbf{1. Steps 1--7:} Algorithm~\ref{TRI} takes as input the bi-connected room adjacency graph $G_B(V_B,E_B)$ produced by Algorithm~\ref{BY}, together with the input non-adjacency constraints (see Figure~\ref{F-3}(a)). The first task is to identify all faces of the input plane graph $G_B$ that are not yet triangles. The procedure \texttt{Find\_Faces$(G_B)$} is called, and its output is stored in the list \texttt{faces}. For each face, the ordered boundary vertices are obtained via \texttt{Vert(face)}. Whenever the boundary contains more than three vertices, and the face is not the outer boundary (checked by \texttt{Is\_Exterior\_Face}), the corresponding vertex list is appended to \texttt{non\_tri\_faces}. In the running example, this produces five non-triangular faces, collected as
\[
\texttt{non\_tri\_faces[0]} = \{10,9,11,12\},\;
\texttt{non\_tri\_faces[1]} = \{0,2,9,10\},\;
\texttt{non\_tri\_faces[2]} = \{9,2,7,11\},
\]
\[
\texttt{non\_tri\_faces[3]} = \{3,5,6,2\},\;
\texttt{non\_tri\_faces[4]} = \{6,5,4,2,1\},
\]
as shown in Figure~\ref{F-3}(b). These faces are exactly the regions that must be triangulated in the subsequent steps.\\\\
\textbf{2. Steps 8--11:} After the list of non-triangular faces has been obtained, an empty set \texttt{Triangulate\_edges} is created to store all new identified edges that will be added further. The function \texttt{ClipTriangulation$(\texttt{non\_tri\_faces},G_B,\texttt{non\_adj\_set})$} is then invoked. It returns a set of interior new edges that triangulate every face in \texttt{non\_tri\_faces} while satisfying the given non-adjacency constraints. These diagonals are stored in \texttt{Triangulate\_edges} and then added to the existing edge set, $E \leftarrow E \cup \texttt{Triangulate\_edges}$. The resulting triangulated graph, denoted $G_T$, is depicted in Figure~\ref{F-3}(h), where the newly introduced edges are highlighted in green.\\\\
\textbf{3. Steps 12--27:}
The function \texttt{Find\_Faces$(G_B)$} computes all faces of a given plane graph $G_B$. 
It begins with a temporary edge set $E'$, which initially contains all edges of $G_B$, and an empty list \texttt{faces}. While $E'$ is not empty, the algorithm selects an edge $s \in E'$ and starts tracing the boundary of a face. 
A list $P$ is created to store the edges of the current cycle, and a pointer $e$ is set to $s$. 
At each step, the function \texttt{Get\_Next\_Edge$(e,G_B)$} returns the next edge in the cyclic order around the current vertex, according to the fixed planar embedding. If the returned edge equals the starting edge $s$, the traversal of the face is complete, and the cycle $P$ is added to \texttt{faces}. Otherwise, the new edge is appended to $P$, the pointer $e$ is updated, and the edge is removed from $E'$. This process continues until all edges have been assigned to boundary cycles. As a result, the algorithm produces one edge-cycle for each face of the planar embedding. 
In later steps, these cycles are filtered to obtain the set \texttt{non\_tri\_faces}.\\\\
\textbf{4. Steps 28--45:} The main triangulation part is performed by function \texttt{ClipTriangulation}, which applies an ear-clipping strategy to each non-triangular face while respecting the non-adjacency constraints provided by user. For every entry in \texttt{non\_tri\_faces}, the ordered boundary is stored in the list \texttt{vertices}. As long as the current face has more than three vertices, the algorithm attempts to clip an ear: a triple of consecutive vertices $(a,b,c)$ for which the diagonal $(a,c)$ is a valid internal diagonal. At the start of each outer iteration, the Boolean variable \texttt{clip\_found} is set to \texttt{false}. The inner loop (Steps 34--38) scans all positions $i$ along the boundary; for a given $i$, the function \texttt{Consecutive$(\texttt{vertices},i)$} returns the triple $(a,b,c)$ consisting of $v_i$, $v_{i+1}$ and $v_{i+2}$ (indices modulo the current boundary length). The candidate diagonal $(a,c)$ is then checked by \texttt{Is\_Valid\_Diagonal}. If the test succeeds, $(a,c)$ is appended to \texttt{new\_edges}, the ear vertex $b$ is removed from \texttt{vertices}, and \texttt{clip\_found} is set to \texttt{true}$;$ the algorithm then proceeds with the remaining smaller polygon. If no valid ear is found in a complete pass (\texttt{clip\_found} remains \texttt{false}), the algorithm clips the first triple returned by \texttt{Consecutive$(\texttt{vertices},0)$} as a fallback (ignoring the non-adjacency constraints), adds the diagonal $(a,c)$, and removes $b$. This process continues until each face is reduced to a single triangle. The diagonals accumulated in \texttt{new\_edges} across all faces form \texttt{Triangulate\_edges}.\\
Figures~\ref{F-3}(c)--(g) illustrate this behaviour: for \texttt{non\_tri\_faces[0]} = $\{10,9,11,12\}$ (Figure~\ref{F-3}(c)), the triple $(10,9,11)$ yields the accepted diagonal $(10,11)$ in the first iteration; for \texttt{non\_tri\_faces[1]}$=$$\{0,2,9,10\}$ (Figure~\ref{F-3}(d)), the candidate $(0,9)$ is rejected in the first iteration (since $(0, 9)$ is part of non-adjacency constraints), whereas $(2,10)$ is accepted in the second iteration; similar clipping steps on the remaining faces (Figures~\ref{F-3}(e)--(g)) produce diagonals/edges such as $(0,11)$, $(5,2)$ and $(1,5)$.  For \texttt{non\_tri\_faces[1]}$=$$\{6,5,4,2,1\}$ (see Figure \ref{F-3}(g)), the possible edge $(5,2)$ is explicitly rejected because it does not form an internal diagonal of the corresponding face, and therefore fails the \texttt{Is\_Valid\_Diagonal} test.\\\\
\textbf{5. Steps 46--60:}
The final two functions describe how candidate edges/diagonals are selected and validated. 
The function \texttt{Consecutive$(\texttt{vertices}, i)$} returns the triple $(a,b,c)$, where $a$ is the vertex at position $i$, $b$ is its successor, and $c$ is its second successor along the current boundary.\\ 
Formally, $a = v_i$, $b = v_{(i+1)\bmod t}$, and $c = v_{(i+2)\bmod t}$ for a boundary of length $t$. The function \texttt{Is\_Valid\_Diagonal$(a,c,\texttt{vertices},face,\texttt{non\_adj\_set})$} checks whether the diagonal $(a,c)$ can be inserted. 
First, the diagonal is rejected if $(a,c)$ appears in the non-adjacency set (see Figure~\ref{F-3}(b), and similar cases in Figure~\ref{F-3}(d)--(f)). 
Next, the algorithm tests whether $(a,c)$ properly intersects any boundary edge $(v_i,v_{i+1})$, excluding the incident edges $(a,b)$ and $(b,c)$. 
If an intersection occurs, the diagonal is rejected. 
Finally, the algorithm verifies that $(a,c)$ is a valid interior diagonal of the polygon and lies strictly inside the face (see Figure~\ref{F-3}(g)). 
The function returns \texttt{true} only if all checks are satisfied.\\
As a result, every diagonal stored in \texttt{Triangulate\_edges} is an interior edge that does not cross existing edges and respects all non-adjacency constraints. 
After inserting these diagonals into $G_B$, the triangulated graph $G_T$ is obtained (Figure~\ref{F-3}(h)), where every bounded face is a triangle. 
This graph is then ready for the subsequent separating-triangle removal step required for rectangular layout construction.\\\\
By applying Algorithm \ref{TRI}, we obtain a bi-connected and triangulated graph $G_T$ derived from the input graph $G_B$ together with the stated non-adjacency constraints (see Figure \ref{F-3}(a–h)). As the input graph is provided with a plane graph $G_B$, every new edge added within a non-triangular face is checked at the time of insertion to avoid any crossings with existing edges (Steps 51-60 of Algorithm \ref{TRI}). This guarantees that planarity is maintained throughout the construction. Therefore, after obtaining the final graph $G_T$, no separate planarity test is necessary, and the algorithm proceeds directly to the subsequent phase of graph generation.

\begin{figure}
\centering
    \includegraphics[width = 0.95 \textwidth]{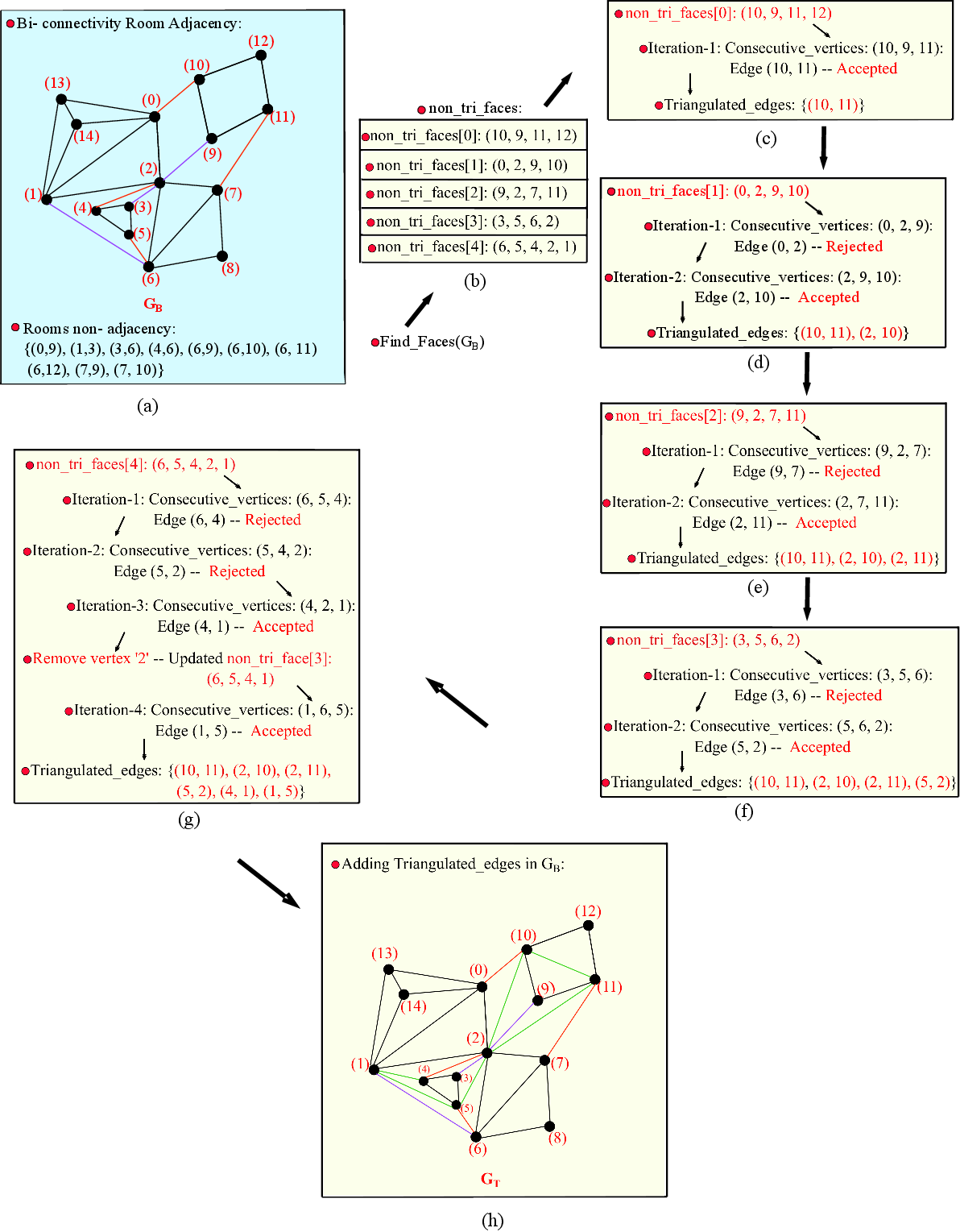}
    \caption{(a-h) Construction of a bi-connected triangulation graph $G_T$ while preserving the specified input constraints.}
    \label{F-3}
\end{figure}
\subsection{Identification and Removal of Separating triangles}\label{Separating}\label{TRI-SEP}
\begin{algorithm}
\caption{Separating triangles removal with non-adjacency}
\label{STR}
\begin{algorithmic}[1]
\Function{$\mathrm{ST\_Edge\_Removal}$}{$G
_T, non\_adj\_set$}
  \State $\mathcal{T}\gets \mathrm{ST}(G_T)$. \Comment{$ST(G_T)$: return all separating Triangles in $G_T$}
  \State $m\gets |\mathcal{T}|$, $E'$ = $E_T$, $G_1(V_1,E_1)$ $=$ $G_T(V_T,E_T)$. 
  \Comment{number of separating triangles}
  \State $(m,G_F)\gets Remove\_ST\_Edge\_Removal (\mathcal{T},G_T,m,E',G_1,|V|,non\_adj\_set).$
  \If{$m>0$}
   \State $E'$ = $\phi$, $(m,G_F)\gets Remove\_ST\_Edge\_Removal(\mathcal{T},G_T,m,E',G_1,|V|,non\_adj\_set).$ 
  \EndIf
  \State \Return \textcolor{red}{$G_F$ ($Free$ $from$ $separating$ $triangle$)}.
\Function{$Remove\_ST\_Edge\_Removal$}{$\mathcal{T},G_T,m,E',G_1,|V|,non\_adj\_set$}

    \ForAll{$T\in \mathcal{T}$}
      \State $removed \gets \text{false}$. \Comment{$removed$ is a boolean variable.}
      \State $E_T \gets \mathrm{Edges}(T)$. 
      \Comment{return edges in separating triangle $T$.}
      \ForAll{edge $e=(x,y)\in E_T$} \Comment{Step 1: try removing an exterior edge}
        \If{$(x,y) \in E'$} \State \textbf{continue} \EndIf
        \State $N_{xy}\gets \mathrm{Common\_Neighbors}(G_T,x,y)$.\Comment{returns common neighbours of $x$ and $y$}
        \If{$|N_{xy}|=2$} \Comment{edge $\{x,y\}$ is an exterior edge.}
          \State $E \gets E$ - $\{(x,y)\}$, $\mathcal{T}\gets \mathrm{ST}(G_T)$, $m\gets m-1$.
          \State $removed \gets \text{true}$.
          \State \textbf{break}
        \EndIf
       \EndFor
      \If{$removed$} \State \textbf{continue} \EndIf
    \Comment{Step 2: try edge replacement $(x,y)\mapsto (r,i)$ respecting $non\_adj\_set$}
      \State $edgeFound \gets \text{false}$.
      \ForAll{edge $e=(x,y)\in E_T$}
       \If{$(x,y) \in E'$} \State \textbf{continue} \Comment{$T$: separating triangle formed by cycle of vertices x, y, z with interior vertices $v_1, ...v_n$.} \EndIf
        \ForAll{$r\in \mathrm{Common\_Neighbors}(G_T,x,y)$}
           \State $d \in \mathrm{nbd}(x)\cap \mathrm{nbd}(y)\cap \mathrm{nbd}(z)$.
          \If{$r\neq d$ \textbf{and} $r\notin T$ \textbf{and} $(r,d) \notin E_T$} 
            \If{$(r,d)\notin \text{non\_adj\_set}$}
              \State $G'_T \gets G_T \cup \{(r,d)\} \setminus \{(x,y)\}$.
               \State $\mathcal{T'}\gets \mathrm{ST}(G'_T)$, $m'\gets |\mathcal{T'}|$.
             \If{$m' \le m-1$ \textbf{and} $\texttt{is\_planar}(G'_T)$}

                \State $G_T \gets G'_T$,  $\mathcal{T} \gets \mathcal{T'}$, $m\gets m'$, $edgeFound \gets \text{true}$.
                \State \textbf{break}.
              \Else
                  \State $G_T \cup \{(x,y)\} \setminus \{(r,d)\}$.
              \EndIf
            \EndIf
          \EndIf
          \EndFor
       \If{$edgeFound$} \State \textbf{break} \EndIf
       \EndFor
      \If{$edgeFound$} \State \textbf{continue} \EndIf
      \Comment{Step 3 (fallback): ignore $non$-$adj$-$constraints$ and retry replacements}
       \ForAll{edge $e=(x,y)\in E_T$}
         \If{$(x,y) \in E'$} \State \textbf{continue}. \Comment{$T$: separating triangle formed by cycle of vertices x, y, z with interior vertices $v_1, ...v_n$.} \EndIf
 \ForAll{$r\in \mathrm{Common\_Neighbors}(G_T,x,y)$}
           \State $d \in \mathrm{nbd}(x)\cap \mathrm{nbd}(y)\cap \mathrm{nbd}(z)$.
          \If{$r\neq d$ \textbf{and} $r\notin T$ \textbf{and} $(r,d) \notin E_T$} 
              \State $G'_T \gets G_T \cup \{(r,d)\} \setminus \{(x,y)\}$.
               \State $\mathcal{T'}\gets \mathrm{ST}(G'_T)$, $m'\gets |\mathcal{T'}|$.
              \If{$m' \le m-1$ \textbf{and} $\texttt{is\_planar}(G'_T)$}
                \State $G_T \gets G'_T$, $\mathcal{T} \gets \mathcal{T'}$, $m\gets m'$, $edgeFound \gets \text{true}$.
                \State \textbf{break}.
              \Else
                  \State $G_T \cup \{(x,y)\} \setminus \{(r,d)\}$.
            
            \EndIf
          \EndIf
      \EndFor
    \EndFor
   \EndFor
  \State \Return $(m,G)$.
\EndFunction
\EndFunction
\end{algorithmic}
\end{algorithm}
\begin{figure}
\centering
    \includegraphics[width = 0.5 \textwidth]{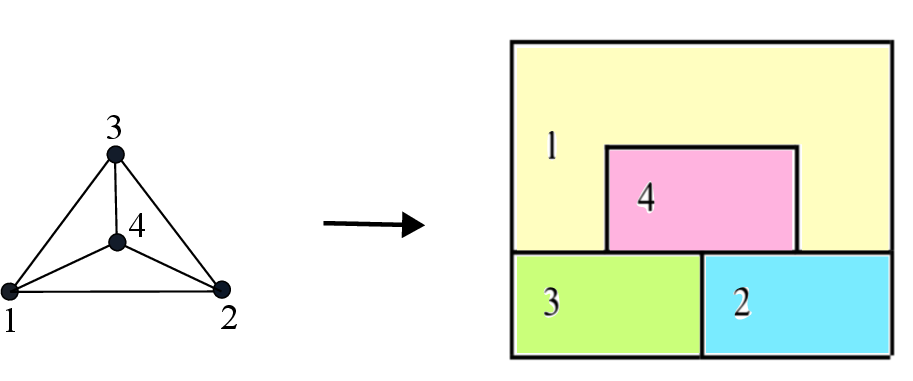}
    \caption{A non-rectangular room (Room 1) appears in the floor plan due to the presence of a separating triangle (1,2,3) in the input graph.}
    \label{EX-2}
\end{figure}
Once a plane-biconnected and triangulated graph $G_T$ has been constructed, the next stage of floor-plan generation depends on the user-specified geometric constraints. 
If no restriction is imposed on module shape (that is, rectilinear shapes such as $L$-, $T$-, or $C$-forms are allowed), the general floor-planning procedure described in Subsection~\ref{OFP} can be applied directly. However, if each module is required to be rectangular, the graph $G_T$ must first be examined for separating triangles. 
Since separating triangles are not compatible with the rectangular floor-plan topology, they must be removed before layout construction. In a rectangular layout, each region is represented by an axis-aligned rectangle with straight and contiguous boundaries. 
Under this representation, three regions cannot form a triangular cycle that encloses another region (see Figure~\ref{EX-2}). It is well established that graphs supporting rectangular floor plans must be free of separating triangles~\cite{rinsma1988existence}.\\
To detect these structures, we apply Johnson’s algorithm~\cite{johnson1975finding} to enumerate all separating triangles, together with the user-defined non-adjacency constraints. 
For each detected triangle, one of its outer edges is selected for removal and replaced with a new diagonal edge. 
This replacement is accepted only if it satisfies the non-adjacency constraints and does not create additional separating triangles. 
If the condition is not satisfied, alternative acceptable edges are considered.

To eliminate a separating triangle $S$ in a graph $G$, we begin by selecting an edge (a separating triangle $S$ is bounded by three edges) incident to $S$ and determining whether this edge lies on the boundary of the outer face of $G$; exterior edges may be removed directly. When all three edges of $S$ are interior, additional structural checks become necessary. The acceptable/candidate edge must be examined to verify that its removal indeed resolves the separating triangle, does not introduce a new separating cycle of the same type, and does not violate any non-adjacency constraints. Algorithm \ref{STR} organizes the decision process into a structured set of cases, each paired with a verified modification. This approach eliminates the need to separate triangles while retaining the graph's structural integrity, which is required for the next stage of layout generation. As illustrated in Figure \ref{F-4} (c-d), resolving the separating triangle (10, 2, 11) involved removing edge (10, 11) and inserting edge (9, 12), thereby preserving triangulation while eliminating the separating structure. Repeating this process yields a set of triangulated graphs free of separating triangles, i.e., PTPG (properly triangulated plane graph) graphs, which then serve as the basis for subsequent rectangular floor plan generation.\\\\
\textbf{Illustration of Algorithm \ref{STR} (see Figure \ref{F-4}):}\\\\
\textbf{1. Steps 1--7:} Algorithm~\ref{STR} starts from the triangulated room adjacency graph $G_T(V_T,E_T)$ returned by Algorithm~\ref{TRI}, together with the room non-adjacency set (Figure~\ref{F-4}(a)). The main function \texttt{ST\_Edge\_Removal} first calls \texttt{ST$(G_T)$} to list all separating triangles of $G_T$; this list is stored in $T$, and its size $m = |T|$ is the current number of separating triangles. The edge set of $G_T$ is copied to a new set $E'_T$, and a reference copy $G_1$ of the original triangulated graph $G_T$ is created. In the running example, \texttt{ST$(G_T)$} returns the five separating triangles shown in Figure~\ref{F-4}(b): $\{0,13,1\}$, $\{10,2,11\}$, $\{2,4,5\}$, $\{1,5,2\}$ and $\{2,1,6\}$. These are the cycles that must be broken in the subsequent steps so that the final graph has no separating triangles. The function \texttt{Remove\_ST\_Edge\_Removal$(T,G_T,m,E'_T,G_1,|V|,\texttt{non\_adj\_set})$} is then called once (and, if needed, a second time) to modify $G_T$ until all these separating triangles are removed. The graph $G_F$ returned by \texttt{ST\_Edge\_Removal} is the separating triangle-free graph used later for rectangular layout construction.\\\\
\textbf{2. Steps 8--22:} Inside \texttt{Remove\_ST\_Edge\_Removal}, the algorithm processes the separating triangles one by one. For a fixed triangle $T\in T$, the flag \texttt{removed} is set to \texttt{false} and the local set $E'_T$ is initialised with the three edges of $T$. The function first attempts a strategy (Step 1 in the comments): delete an exterior edge of $T$ without adding any new edges. To do this, it scans all edges $e=(x,y)$ of $G_T$ and only considers those belonging to $E'_T$. For each such edge, it computes the common neighbours $N_{xy}$ of $x$ and $y$ in $G_T$. If $|N_{xy}| = 2$, the edge $(x,y)$ is incident to exactly two triangular faces and is therefore an exterior edge of $T$; in that case the edge is removed from $E_T$, the list of separating triangles is recomputed as $T \leftarrow \texttt{ST}(G_T)$, the counter $m$ is replaced by the new value $m' = |T|$, and \texttt{removed} is set to \texttt{true}. The algorithm then goes directly to the next triangle. Figure~\ref{F-4}(c) illustrates this case for $\texttt{Sep\_Tri[0]} = \{0,13,1\}$: the edge $(0,13)$ is recognised as an exterior edge and deleted, which removes the separating triangle without introducing any new one.\\\\
\textbf{3. Steps 23--41:} If no exterior edge of a separating triangle $T$ can be safely removed (that is, the flag \texttt{removed} remains \texttt{false}), the algorithm proceeds to a more general resolution strategy (Step~2), in which an edge of $T$ is replaced by a new edge while still respecting the given non-adjacency constraints. A boolean flag \texttt{edgeFound} is initialized to \texttt{false}. The algorithm then considers each edge $e=(x,y)$ of $T$ as a potential candidate edge for replacement. For a selected edge $(x,y)$, it first identifies vertices $r$ that are common neighbours of $x$ and $y$, and then identifies vertices $d$ that are adjacent to $x$, $y$, and the third vertex $z$ of $T$; such vertices $d$ lie in the interior region enclosed by $T$. Whenever $r \neq d$, $r$ is not a vertex of $T$, and the edge $(r,d)$ is not already present in $E_T$, the pair $(r,d)$ is considered as a possible replacement for $(x,y)$.\\
Each candidate edge $(r,d)$ is first checked against the non-adjacency constraints and discarded if it violates any user-specified restriction. Otherwise, a temporary graph $G'_T$ is constructed by inserting $(r,d)$ and removing $(x,y)$. The set of separating triangles in $G'_T$ is then recomputed as $T' = \texttt{ST}(G'_T)$, with $m' = |T'|$. If this modification strictly reduces the number of separating triangles ($m' \le m - 1$), the replacement is accepted: $G_T$ is updated to $G'_T$, the current triangle set is replaced by $T'$, $m$ is updated to $m'$, and \texttt{edgeFound} is set to \texttt{true}, allowing the algorithm to continue with the next separating triangle.\\
Figure~\ref{F-4}(d) illustrates such a successful constrained replacement for $\texttt{Sep\_Tri}[1] = \{10,2,11\}$. In this example, several candidate replacements involving edges incident to $(10,2)$ and $(2,11)$ are rejected until the algorithm removes the edge $(11,10)$ and inserts $(9,12)$. This modification respects the non-adjacency constraints and eliminates the separating triangle without introducing new ones.\\\\
\textbf{4. Steps 42--54:} The final stage of the procedure introduces a fallback strategy (Step~3), which is invoked only when Step~2 fails to identify any admissible edge replacement that satisfies the non-adjacency constraints, that is, when \texttt{edgeFound} remains \texttt{false} after the constrained search. In this stage, the algorithm repeats the examination of triangle edges and candidate vertex pairs, but the non-adjacency check $(r,d) \in \texttt{non\_adj\_set}$ is deliberately omitted. For each admissible pair $(r,d)$ with $r \neq d$, $r \notin T$, and $(r,d) \notin E_T$, a temporary graph $G'_T$ is formed by inserting $(r,d)$ and removing $(x,y)$. The separating triangles of $G'_T$ are then recomputed as $T' = \texttt{ST}(G'_T)$, with $m' = |T'|$. If this modification strictly reduces the number of separating triangles ($m' \le m - 1$), the replacement is accepted following the same update procedure used in Step~2.\\
This fallback mechanism guarantees progress even when all admissible replacements are excluded by non-adjacency constraints. Its effect is illustrated in Figures~\ref{F-4}(e) and~\ref{F-4}(f). For the separating triangles $\{2,4,5\}$ and $\{1,5,2\}$, every constraint-respecting candidate is rejected, prompting Step~3 to bypass the non-adjacency set and replace the edge $(5,2)$ with $(3,6)$, thereby eliminating both triangles simultaneously (Figure~\ref{F-4}(e)). For the remaining triangle $\{2,1,6\}$ (Figure~\ref{F-4}(f)), the algorithm identifies $(1,6)$ as an exterior edge and removes it, eliminating the final separating triangle. After all triangles in $T$ have been processed, the graph $G_T$ evolves into the final graph $G_F$ shown in Figure~\ref{F-4}(f), which contains no separating triangles. Algorithm~\ref{STR} therefore produces a triangulated, non-separating graph that respects non-adjacency constraints whenever possible and serves as a valid combinatorial foundation for the subsequent rectangular layout construction.\\\\
Applying Algorithm~\ref{STR} transforms the triangulated input graph $G_T$, together with the specified non-adjacency constraints, into a properly triangulated plane graph $G_F$ (see Figure~\ref{F-4}(a--f)). Since $G_T$ is given with a fixed plane embedding, any edge introduced during the elimination of separating triangles is inserted only within non-triangular faces created by earlier edge removals. Each such insertion is verified at the time of addition to ensure that it does not intersect existing edges, thereby preserving planarity throughout the process. After construction terminates, the resulting graph $G_F$ admits a planar embedding, which is used directly as the combinatorial input for the subsequent floor-plan generation phase.

\begin{figure}
\centering
    \includegraphics[width = 0.95 \textwidth]{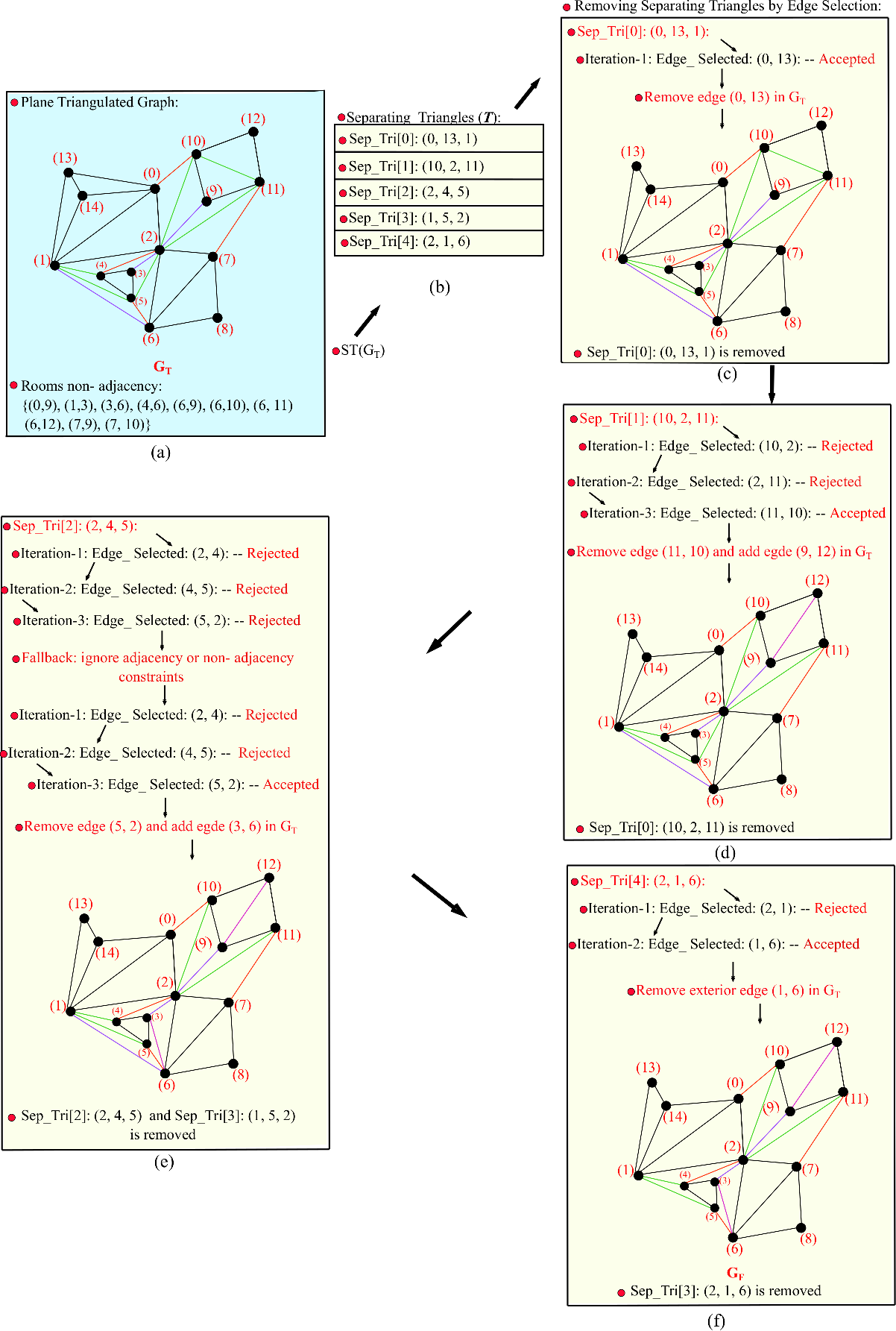}
    \caption{(a-f) Breaking of separating triangles.}
    \label{F-4}
\end{figure}

\subsection{Graph to Floorplan Generation}
\label{PTPG}
\subsubsection{RFP Generation:}

After removing all separating triangles from the triangulated graph $G_T$ (see Subsection~\ref{Separating}), the resulting graph $G_F$ is a properly triangulated plane graph (PTPG) that satisfies the user input constraints. 
Such graphs meet the structural conditions required for rectangular dual construction, as discussed by Koźmiński et al.~\cite{kozminski1985rectangular}. 
These PTPGs therefore serve as input for rectangular floor-plan generation.\\
A \emph{dimensionless floor plan} (UFP) is a layout in which each room is modeled as an axis-aligned rectangle, and only adjacency relations between rectangles are fixed; exact dimensions and areas are not assigned. This abstraction allows designers to examine circulation patterns, relative room placement, and orientation without committing to precise measurements. It also supports comparing alternative layouts and detecting adjacency conflicts early, before geometric refinement.\\
Graphs provide a clear way to represent which rooms must be adjacent and which rooms must remain separate. Therefore, graph-based approaches are well-suited for generating families of dimensionless floor plans. In our implementation, the PTPG constructed in Section~\ref{PTPG} is used as input to the algorithm of Kant et al.~\cite{kant1997regular}, which produces a rectangular representation where each rectangle corresponds to a room.\\
The resulting UFP serves as a structural blueprint that can later be scaled and detailed. 
For the example PTPG shown in Figure~\ref{F-4}f, the corresponding UFPs are presented in Figure~\ref{F-5}. 
Although these layouts share the same underlying adjacency graph, they differ in geometric arrangement. 
For instance, rooms 4 and 1 share a vertical wall segment in one layout and a horizontal wall segment in another. 
More generally, the alternatives vary in relative room placement and shared wall segments while preserving the same adjacency structure.
\begin{figure}
\centering
    \includegraphics[width = 0.95 \textwidth]{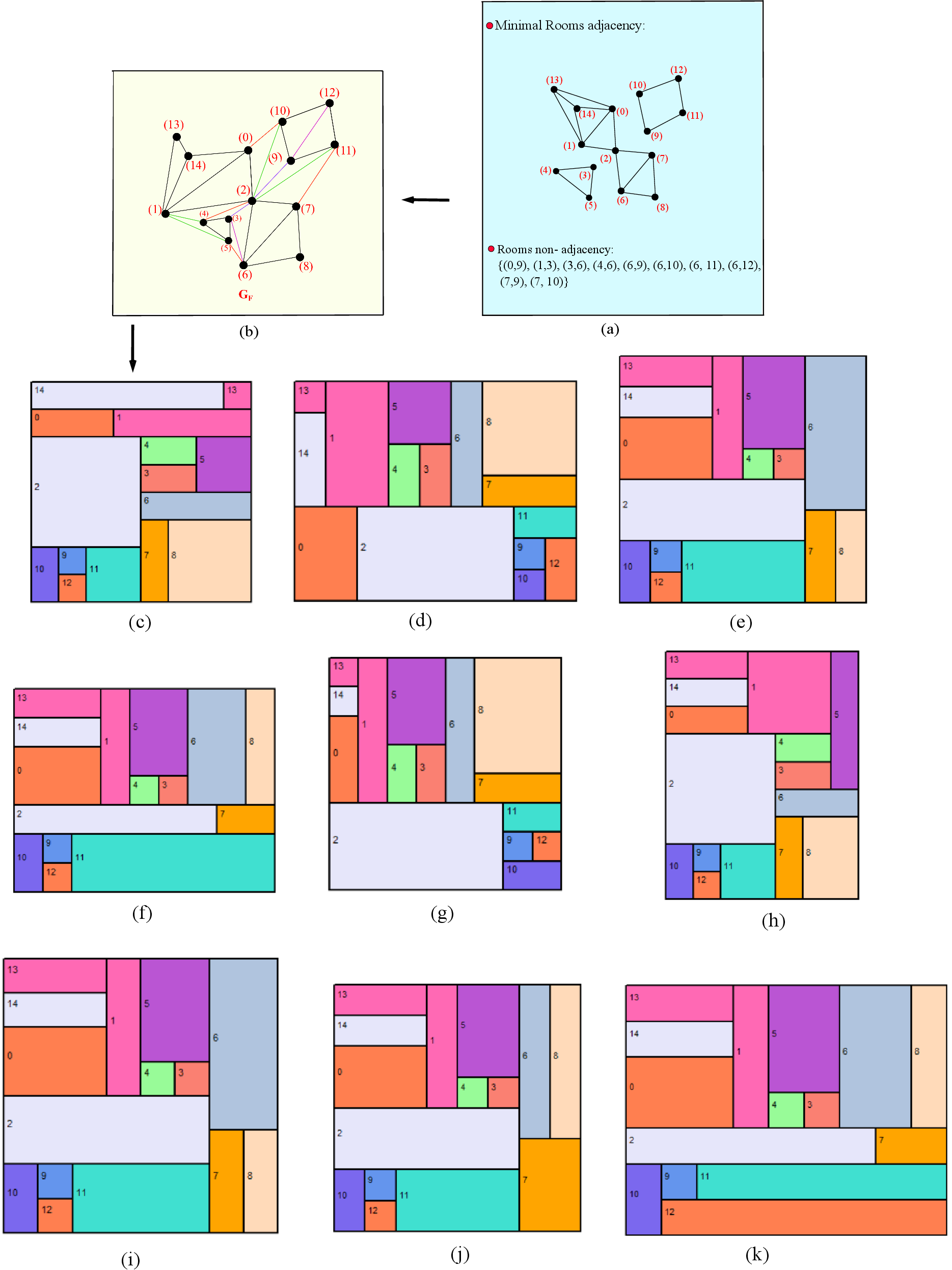}
    \caption{(a-k) Multiple Rectangular Floor plans generated based on the previously generated graph $G_F$.}
    \label{F-5}
\end{figure}

\subsubsection{OFP Generation:}\label{OFP}
In our proposed prototype, designers and architects can generate either rectangular floor plans (RFPs) or orthogonal floor plans (OFPs). When the OFP option is selected, modules can take general rectilinear shapes rather than being restricted to rectangles. See Figure \ref{F-6}: Starting from the generated graph $G_T$, the construction follows the methodology described by Krishnendra et al. \cite{SHEKHAWAT2021103718}. As established by the theorem of Kant and He \cite{kant1997regular}, a bi-connected planar triangulated graph containing a separating triangle cannot admit an RFP. Accordingly, in such cases, a generated PTPG $G_T$ is augmented by introducing auxiliary vertices and further triangulated so that the resulting modified graph is free of separating triangles. An RFP is then constructed for this modified graph, and the regions corresponding to the auxiliary vertices are subsequently merged with suitable neighboring rooms to obtain an OFP for the original $G_T$. This sequence of transformations, illustrated in Figure \ref{F-6} (a-d), systematically removes separating triangles while preserving planarity and adjacency constraints. Furthermore, since a single modified graph may yield multiple valid RFP realizations, selectively merging auxiliary modules enables the generation of multiple distinct OFPs, as shown in Figure \ref{F-6} (e-m).
\begin{figure}
\centering
    \includegraphics[width = 0.91 \textwidth]{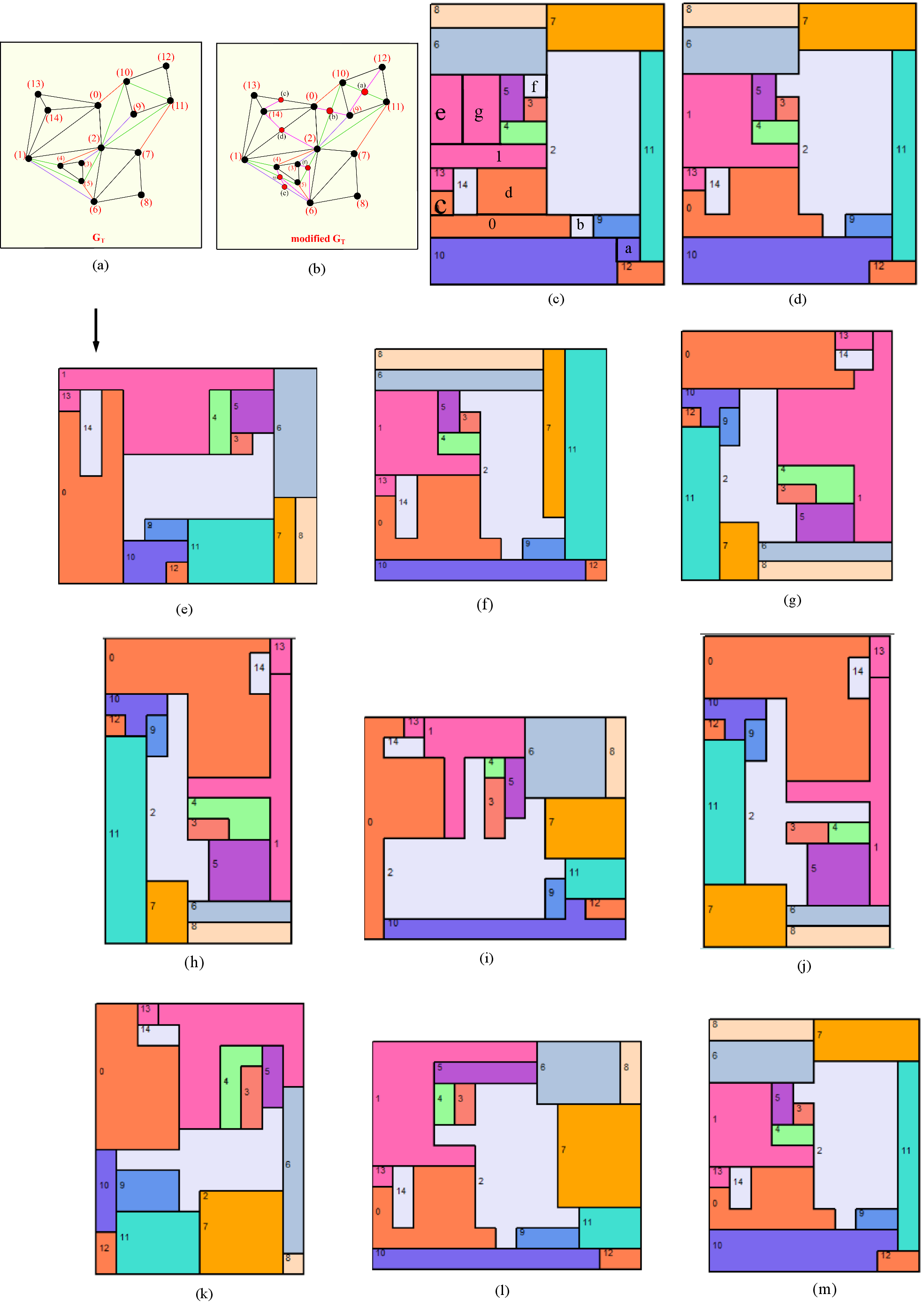}
    \caption{(a-m) Multiple orthogonal Floor plans generated based on the previously generated graph $G_T$.}
    \label{F-6}
\end{figure}
\section{Practical Use of the DPLAN Interface}\label{Apply}
This section illustrates the practical use of the proposed DPLAN prototype by highlighting how automatically generated residential floor plans/layouts can be refined through user-driven customization. After the system produces an initial rectangular layout, users can modify the layout boundaries and room geometries to accommodate individual design preferences and functional constraints, including creating non-rectangular spaces that extend beyond standard configurations. Figures \ref{F-8} and \ref{F-9} present representative examples of two- and three-bedroom residential layouts, respectively, demonstrating how variations in input constraints, such as adjacency and non-adjacency relations between rooms, as well as room shape specifications, yield distinct customized outcomes. These examples highlight the precision enabled by manual refinement and demonstrate how the graph-generation design framework embedded in the software facilitates the generation of realistic, adaptable floor plans. Together, they demonstrate the effectiveness of the proposed approach in producing residential layouts that closely align with user requirements and real-world design considerations.
\begin{figure}
\centering
    \includegraphics[width = 0.95 \textwidth]{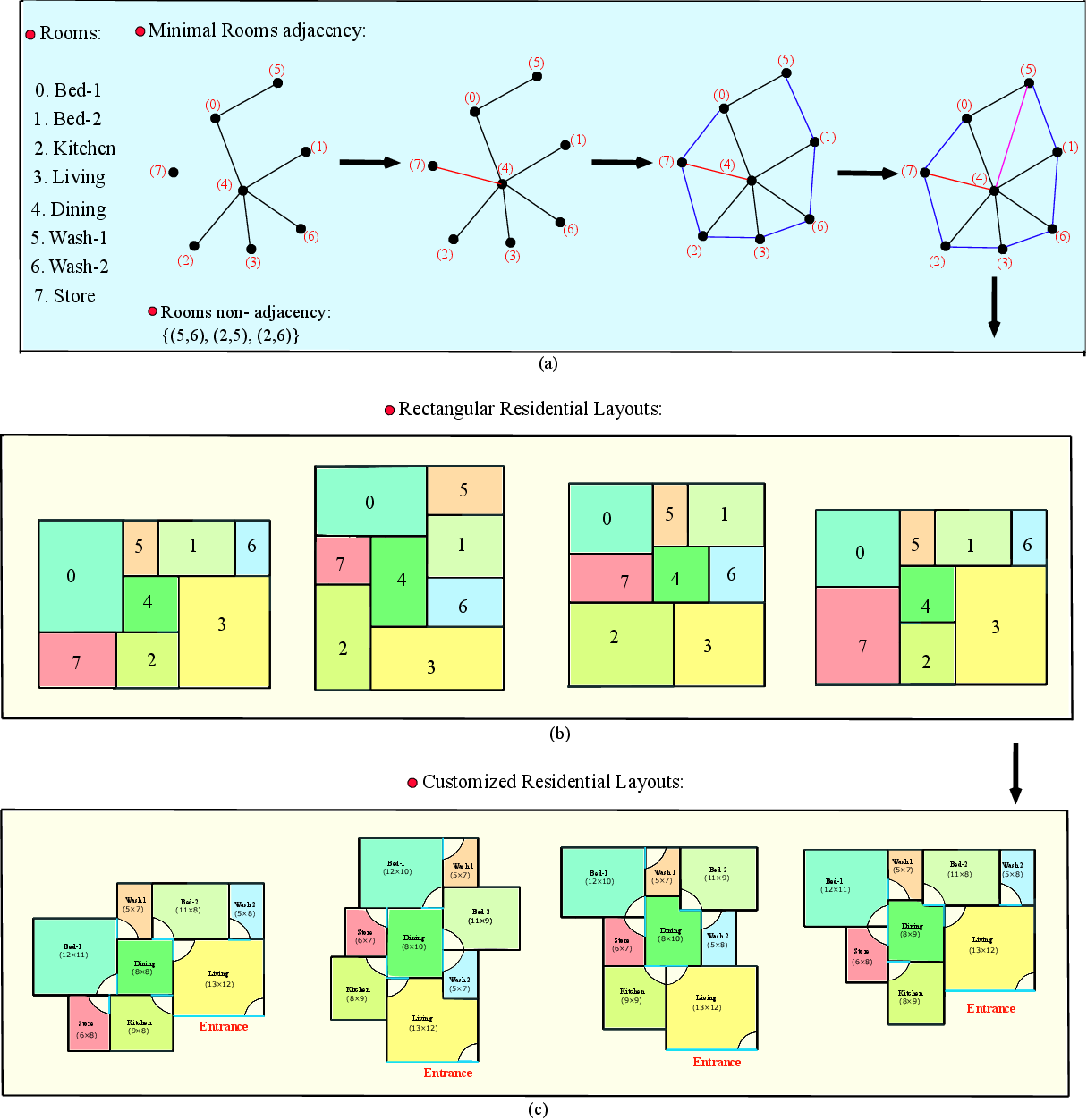}
    \caption{(a–c) Customized two-bedroom residential floor plans obtained through user-guided refinement of the generated layouts.}
    \label{F-8}
\end{figure}
\begin{figure}
\centering
    \includegraphics[width = 1.00 \textwidth]{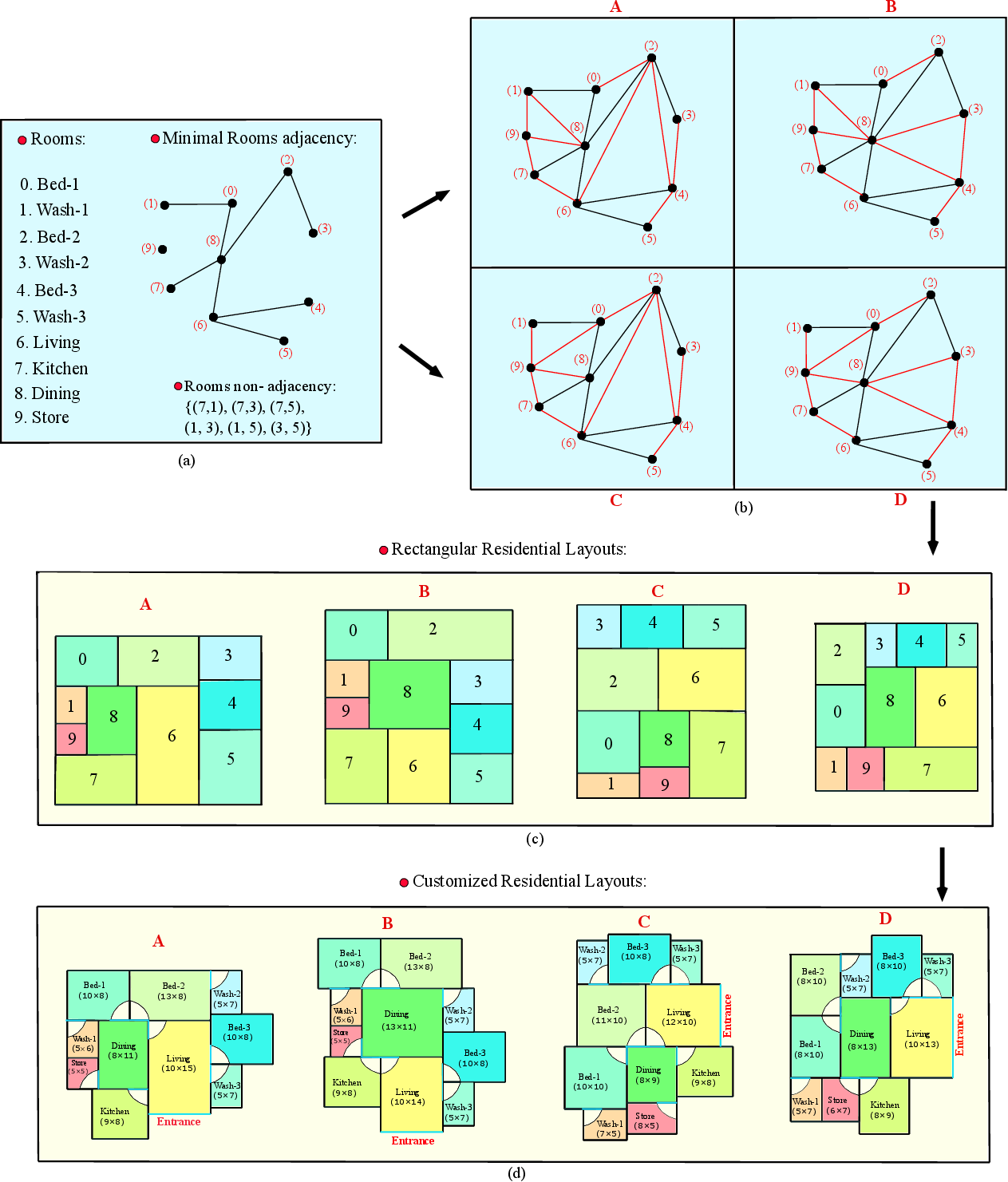}
    \caption{(a–d) Customized three-bedroom residential floor plans obtained through user-guided refinement of the generated layouts.}
    \label{F-9}
\end{figure}

\section{Conclusion}\label{conclusion}
This work introduces a DPLAN prototype that applies a sequence of graph-based algorithms to generate families of floor plans, either rectangular floor plans (RFPs) or orthogonal floor plans (OFPs), from user-defined constraints. 
The user provides high-level spatial requirements, including which rooms must be adjacent (for possible door placement) and which room pairs must remain non-adjacent and not share a wall. Based on these inputs, the system constructs a biconnected, plane-triangulated graph that satisfies both adjacency and non-adjacency constraints. This properly triangulated plane graph (PTPG) acts as the combinatorial foundation for producing multiple, distinct, yet constraint-consistent layout alternatives. At its current stage, the prototype generates layouts with a rectangular outer boundary and produces dimensionless (topological) floor plans suitable for early design exploration and structural analysis.

\section{Future Enhancement}\label{FE}

Several extensions are planned to bring the prototype closer to practical architectural use. 
On the geometric side, future work will support non-rectangular outer boundaries and assign explicit dimensions to rooms and corridors. 
This will allow the system to produce scale-accurate floor plans rather than purely combinatorial layouts. Circulation modelling will also be improved by enabling users to define corridor requirements and preferred access patterns, which can be incorporated directly into the graph constraints. 
From a usability perspective, the interface will be enhanced to simplify constraint editing, comparison of alternative layouts, and export to standard CAD or BIM platforms. Additional developments include automatic door placement, basic three-dimensional visualization of generated layouts, and integration with evaluation modules such as daylight analysis or area-efficiency assessment. 
Together, these extensions aim to develop the prototype into a more complete decision-support tool for architects and space planners.
\section{Supplementary Material}
\label{Supply}

A video illustrating the execution of the proposed Python-based framework, which generates floor plans under different user-defined constraints, is available at the following link:\\
\href{https://drive.google.com/file/d/1mpvako-JXxCN0JERFrzHxO0y4ojeE8W-/view?usp=sharing}
{Demonstration}\par

\noindent A brief demonstration of our updated prototype, reflecting ongoing developments, is provided below. This video highlights the revised editing workflow along with recent usability enhancements:\\
\href{https://drive.google.com/file/d/1PtSCOteH9kqNrCggafm6MAglesn6Wvxk/view?usp=drive_link}{Future Work} \par

\section{Declaration of Interest}
The authors confirm that no financial or personal relationships exist that could be perceived as influencing the results reported in this paper.

\bibliographystyle{unsrt} 
\bibliography{mybib} 
\section{Appendix}\label{App}
\subsection{Analysis of Performance}
\label{sec:performance}

This section explains the computational performance of our proposed pipeline and compares it with existing graph-based and learning-based methods for floorplan generation. We analyze each step of the algorithm separately and describe the type of structural operations performed at that stage. Using the notation and helper functions defined earlier in Section~\ref{Meth}, we derive clear worst-case time complexity bounds for each step.\\
We also discuss how these computational costs relate to the properties of typical architectural room adjacency graphs, such as their size and structural constraints. Finally, we compare our method with established approaches to show where our contribution fits within existing architectural layout generation techniques.

\subsubsection{Overview of the Computational Structure}

The pipeline described in Section~\ref{Meth} converts a user-defined door-adjacency graph \( G(V, E) \) into a properly triangulated plane graph \( G_F \). This transformation is carried out in four structured stages, each enforcing a specific graph-theoretic property required for the final layout/floor plan.

\begin{enumerate}
    \item Algorithm~\ref{OF-CON} ensures connectivity while respecting non-adjacency constraints, producing the graph \( G_C \).
    
    \item Algorithm~\ref{BY} augments the graph to achieve bi-connectivity, producing \( G_B \).
    
    \item Algorithm~\ref{TRI} performs constrained triangulation to obtain \( G_T \).
    
    \item Algorithm~\ref{STR} removes separating triangles (when required), resulting in the final graph \( G_F \), particularly when each module must correspond to a rectangular region.
\end{enumerate}

All four procedures operate directly on a fixed plane embedding of the graph. At every step, adjacency relations are modified only when necessary and in a clearly defined manner. The pipeline does not explore or enumerate multiple alternative adjacency configurations, unlike learning-based or search-based approaches. Instead, it follows a deterministic sequence of structural modifications, where each step serves a single, well-defined purpose: enforcing connectivity, bi-connectivity, triangulation, or eliminating separating triangles.\\
Therefore, the computational cost of the pipeline depends only on the size and structural properties of the input graph. It does not depend on exploring multiple combinatorial possibilities, which keeps the behavior predictable and easier to analyze.
\subsubsection{Run-time Analysis of the Proposed Algorithms}

We now analyze the worst-case time complexity of Algorithms~\ref{OF-CON}--\ref{STR}. 
Let \( n = |V| \) denote the number of rooms (vertices), and let \( m = |E| \) denote the number of edges in the current adjacency graph. Since the input graph is planar, we have \( m = O(n) \) at every stage of the pipeline.\\
The non-adjacency constraints are provided as part of the input and stored in a data structure that allows constant-time access for any specified pair of rooms. This assumption applies only to retrieving the constraint information itself.
During the execution of Algorithms~\ref{OF-CON}--\ref{STR}, these non-adjacency constraints are checked repeatedly whenever new possible/candidate edges are generated and evaluated. Although each individual lookup takes constant time, the total cost of enforcing non-adjacency constraints depends on the number of candidate edge combinations examined. Therefore, the overall contribution of constraint checking is naturally included in the total time complexity of the corresponding algorithmic steps.
\subsubsection{Algorithm \ref{OF-CON}: Connectivity with Non-Adjacency:}
Algorithm \ref{OF-CON} ensures that every room/node lies in a single connected component while satisfying the non-adjacency constraints. Its complexity arises from three primary operations:

\begin{itemize}
    \item \textbf{Connected components:} The call \textsc{Connected\_Components}$(G)$ runs in $O(n+m)$ time using the standard Depth First Search (DFS) algorithm discussed in the paper \cite{korach1988dfs}, which takes linear time complexity.\\
    \item \textbf{Extraction of outerface boundaries:}  For each disconnected component formed during connectivity augmentation, \textsc{Extract\_Outerface}$()$ function traces the boundary of the outer face in the plane embedding to identify its incident vertices. As outlined in Algorithm \ref{OF-CON}, this traversal follows the cyclic edge order around boundary vertices and marks visited edges and vertices to prevent repetition. Since each vertex and edge is encountered at most once across all components, the total cost of outerface extraction is linear in the size of the graph. Combined with the connected component computation and union-find operations, this step contributes $O(n+m)$ time to the overall complexity, where $n$ and $m$ denote the number of vertices and edges, respectively.\\
\item \textbf{Candidate edge generation:} 
After computing the outer-face vertex sets for all connected components, the function \textsc{Generate\_Possible\_Edges} systematically constructs admissible/possible edges that may be used to connect distinct components without violating non-adjacency constraints. Let the input graph consist of $k$ connected components, and let $b_i$ denote the number of vertices incident to the outer face of component $i$. Since only boundary vertices can be connected without destroying planarity, the algorithm restricts its search to cross-component pairs drawn from these outer-face sets. \\
For each unordered pair of components $(i,j)$ with $i<j$, the algorithm examines all vertex pairs $(u,v)$ such that $u \in \text{outerfaces}[i]$ and $v \in \text{outerfaces}[j]$. The total number of candidate pairs considered is therefore
\[
\sum_{i<j} b_i b_j
\;=\; \frac12\!\left(\left(\sum_{i=1}^k b_i\right)^2 - \sum_{i=1}^k b_i^2\right).
\]
Because every outer-face vertex belongs to exactly one component and the union of all outer faces is a subset of the vertex set $V$, we have $\sum_i b_i \le n$. Consequently, the total number of vertex pairs examined is bounded above by $O(n^2)$ in the worst case.\\
Each candidate pair is cross-checked against the Room Non-Adjacency list. This test determines whether adding the corresponding edge would violate a user-specified separation constraint. The non-adjacency set is stored in a hash-based data structure, so each lookup runs in constant expected time. Importantly, this constant-time check applies only to verifying whether a specific pair $(u,v)$ is forbidden; it does not imply that the overall edge-generation process is constant time, as the algorithm must still enumerate all admissible cross-component pairs.\\
Combining the quadratic bound on the number of candidate pairs with the constant-time feasibility check for each pair, the overall time complexity of \textsc{Generate\_Possible\_Edges} is $O(n^2)$ in the worst case. \\
\item \textbf{Union find merging:} 
Once the set of admissible candidate edges has been generated, Algorithm \ref{OF-CON} iterates through this set to progressively merge disconnected components. Each candidate edge $(u,v)$ connects vertices belonging to two distinct components, as determined by their representative elements in the disjoint-set structure. Before performing a merge, the algorithm queries the union-find data structure to check whether $u$ and $v$ already belong to the same component; if not, it performs a union operation to merge the corresponding components.\\
The number of such union--find operations is bounded by the number of candidate edges examined. As established in the candidate edge generation phase, the total number of admissible cross-component vertex pairs is $O(n^2)$ in the worst case. As a result, the total number of find and union operations executed by the augmentation loop throughout Algorithm \ref{OF-CON} is bounded by $O(n^2)$.\\\\
Each connectivity test and component merge in Algorithm \ref{OF-CON} is performed using a union-find data structure augmented with path compression and union-by-rank (or size). With these heuristics, the total cost of performing a sequence of union and find operations on $n$ elements is bounded by $O(p\,\alpha(n))$, where $p$ is the number of operations and $\alpha(\cdot)$ is the inverse Ackermann function \cite{tarjan1975efficiency}. This function grows extremely slowly and remains bounded by a small constant for all input sizes relevant to architectural layout graphs. As a result, although Algorithm~\ref{OF-CON} may invoke union find operations for up to $O(n^2)$ candidate edges, the cost of each operation is negligible compared to the cost of enumerating the candidates themselves. Consequently, the union--find merging phase does not increase the overall asymptotic complexity beyond $O(n^2)$.\\
Combining the quadratic bound on the number of candidate edges with the near-constant amortized cost of each union-find operation, the total time spent in the merging phase of Algorithm \ref{OF-CON} is bounded by $O(n^2)$. This bound reflects the algorithm’s deliberate strategy of exhaustively testing all planarity-preserving, non-adjacent cross-component connections while maintaining efficient component tracking through a disjoint-set data structure.
\end{itemize}
Thus, the total running time of Algorithm \ref{OF-CON} is
\[
T_{\mathrm{conn}} = O(n^2),
\]
with the quadratic term reflecting the necessity of comparing boundary vertices across disconnected components.

\subsubsection{Algorithm~\ref{BY}: Bi-Connectivity with Non-Adjacency:}

Algorithm~\ref{BY} upgrades the connected graph $G_C$ to a bi-connected graph $G_B$ by eliminating articulation vertices, that is, vertices whose removal disconnects the current door-connectivity structure. The algorithm operates by locally repairing articulation-induced separations while respecting the Room Non-Adjacency constraints, and continues until no articulation vertex remains.

\begin{itemize}
    \item \textbf{Block decomposition around an articulation vertex:}
    For a fixed articulation vertex $v$, the procedure \textsc{Blocks}$(G_C,v)$ constructs the graph obtained by removing $v$ and its incident edges, and then computes the connected components of the resulting graph. Using a standard DFS/BFS traversal \cite{korach1988dfs}, this step requires $O(n+m)$ time on the current graph. Since the graph remains planar throughout the pipeline, we have $m=O(n)$, and hence block decomposition is linear in $n$.\\
\item \textbf{Cross-block candidate enumeration under non-adjacency:}
    Let $\{B_1,\dots,B_b\}$ denote the blocks incident to $v$. The routine \textsc{Valid\_Edges} enumerates all vertex pairs $(u,w)$ with $u\in B_i$ and $w\in B_j$ for all unordered block pairs $(B_i,B_j)$, and filters them by testing whether $(u,w)\notin\texttt{non\_adj\_set}$. The number of candidate pairs examined is
    \[
        \sum_{i<j} |B_i|\,|B_j|
        \;=\;
        \tfrac12\!\left(\left(\sum_i |B_i|\right)^2 - \sum_i |B_i|^2\right)
        \;\le\; O(n^2),
    \]
    since the blocks form a partition of $V\setminus\{v\}$. When the Room Non-Adjacency relation is stored in a representation supporting constant-time membership queries, the total cost of this enumeration step is $O(n^2)$ in the worst case.\\
\item \textbf{Merging blocks via union--find:}
    From the admissible candidate edges, the procedure \textsc{Connect\_Blocks} selects a minimal subset that merges all blocks incident to $v$ into a single bi-connected structure. This is implemented using a disjoint-set data structure over block indices. Each candidate edge induces a constant number of \textsc{Find} and \textsc{Union} operations. With path compression and union-by-rank, a sequence of $p$ such operations runs in $O(p\,\alpha(n))$ amortised time. Here, $p=O(n^2)$ in the worst case, so this step is asymptotically dominated by the candidate enumeration cost.
\end{itemize}

\noindent
\textbf{Remark on edge insertions and global progress:}
Although Algorithm~\ref{BY} inserts edges during execution, the graph remains planar at all times; therefore, the total number of edges satisfies $m\le 3n-6$, and all subroutines of cost $O(n+m)$ remain linear in $n$. Moreover, each accepted edge insertion strictly reduces the articulation structure of the graph by merging blocks incident to the processed articulation vertex. Consequently, the total number of accepted insertions over the entire execution is $O(n)$.\\\\
While the above steps are described locally for a single articulation vertex, the algorithm does not incur the worst-case $O(n^2)$ cost independently for each such vertex. Because each insertion reduces the number of articulation-induced blocks and no articulation vertex is processed indefinitely, the total number of cross-block vertex pairs examined across the full execution is bounded by $O(n^2)$. 

\noindent
Therefore, the overall worst-case running time of Algorithm~\ref{BY} is
\[
    T_{\mathrm{biconn}} = O(n^2),
\]
with the quadratic term arising from explicit cross-block candidate enumeration under non-adjacency constraints. In typical architectural door graphs, the number of articulation vertices and the sizes of block partitions are small, and the observed running time is substantially lower than this worst-case bound.

\subsubsection{Algorithm~\ref{TRI}: Constrained Triangulation:}

Algorithm~\ref{TRI} transforms the bi-connected planar graph $G_B$ into a plane triangulation $G_T$ while respecting the Room Non-Adjacency constraints whenever feasible. The procedure operates on the planar embedding of $G_B$ and consists of enumerating all faces and triangulating each non-triangular interior face by inserting admissible diagonals.

\begin{itemize}
    \item \textbf{Face enumeration:}
    The routine \textsc{Find\_Faces}$(G_B)$ extracts all facial cycles of the planar embedding by traversing directed edges according to their cyclic order around vertices. Each directed edge is marked once it has been assigned to a face, and since each edge is incident to at most two faces, the total cost of enumerating all faces is $O(n+m)$. As the graph remains planar throughout, $m=O(n)$ and face enumeration is linear in $n$.\\
\item \textbf{Amortised ear-clipping triangulation \cite{mei2012ear}:}
    Let $f$ be an interior face whose boundary is a simple polygon of length $\ell>3$. The procedure \textsc{Clip\_Triangulation} repeatedly removes boundary vertices by inserting diagonals until $f$ is decomposed into triangles, performing exactly $\ell-3$ successful clipping steps.

    In each step, the algorithm scans the current boundary to identify a triple of consecutive vertices $(a,b,c)$ such that the diagonal $(a,c)$ satisfies the predicate \textsc{Is\_Valid\_Diagonal}. This predicate includes (i) a membership query against the Room Non-Adjacency set, (ii) a geometric intersection test against existing boundary edges, and (iii) a check that the diagonal lies in the interior of the face. In the straightforward implementation, the geometric intersection test dominates and requires $O(\ell)$ time.

    Importantly, each diagonal candidate is tested at most a constant number of times before either being accepted or becoming irrelevant due to boundary reduction. As a result, the total cost of all diagonal validity tests over the full triangulation of a face of length $\ell$ is $O(\ell^2)$. 
\end{itemize}

\noindent
To obtain a graph-level bound, observe that the sum of boundary lengths over all faces in a planar embedding is $O(m)$, and since $m=O(n)$, the total triangulation cost over all faces is
\[
T_{\mathrm{tri}} = \sum_{\text{faces}} O(\ell^2) = O(n^2).
\]

\noindent
In practice, the non-triangular faces arising from architectural door-connectivity graphs are typically small, and admissible diagonals are identified quickly, leading to substantially faster observed runtimes.

\subsubsection{Algorithm~\ref{STR}: Separating-Triangle Removal}

Algorithm~\ref{STR} eliminates separating triangles from the triangulated graph $G_T$ in order to obtain a properly triangulated plane graph $G_F$ (PTPG), which forms the required input class for rectangular-dual based layout construction~\cite{he1999floor}. The algorithm processes separating triangles iteratively and applies local graph modifications that strictly reduce their number, while preserving planarity and triangulation at every step.

\begin{itemize}
    \item \textbf{Detection of separating triangles \cite{johnson1975finding}:}
    The procedure $\textsc{ST}(G_T)$ identifies all separating triangles in the current graph. Since $G_T$ is a plane triangulation, all triangular faces can be enumerated by adjacency-based traversal, and each candidate triangle can be tested for separability by checking whether it encloses interior vertices. This detection step runs in $O(n+m)$ time, and because planarity is maintained throughout the algorithm, $m=O(n)$, yielding linear-time detection.\\
\item \textbf{Exterior-edge deletion test:}
    For a separating triangle $T=(x,y,z)$, Algorithm~\ref{STR} first attempts to delete an exterior edge. For an edge $(x,y)\in E(T)$, the algorithm computes the common neighbours of $x$ and $y$. In a plane triangulation, each edge is incident to exactly two triangular faces, so the number of common neighbours is constant. Consequently, testing whether an edge is exterior and deletable requires only local adjacency checks and incurs constant time per edge. Since each separating triangle has exactly three edges, this stage incurs a cost of $O(1)$ per triangle.\\
    \item \textbf{Constrained edge replacement.}
    If no exterior edge can be deleted, the algorithm attempts to replace an edge of the separating triangle while preserving triangulation and respecting non-adjacency constraints. Let $T=(x,y,z)$ be a separating triangle. For each boundary edge $(x,y)$, the algorithm selects a vertex $r$ from $\textsc{CommonNeighbors}(x,y)$ and a vertex $d$ from $\mathrm{nbd}(x)\cap\mathrm{nbd}(y)\cap\mathrm{nbd}(z)$. In a plane triangulation, each edge has exactly two face-completing common neighbours, and a triangle admits at most one vertex adjacent to all three of its corners. Hence, the number of candidate $(r,d)$ pairs examined per separating triangle is bounded by a constant.\\\\
    For each candidate replacement, the algorithm performs: (i) a membership query against the Room Non-Adjacency set, (ii) a constant-time local edge update that preserves planarity, and (iii) a recomputation of the separating triangles in the modified graph, denoted by $\textsc{ST}(G_T')$. The recomputation step dominates and requires $O(n)$ time.\\
    \item \textbf{Fallback replacement without non-adjacency filtering.}
    If no admissible replacement satisfying the non-adjacency constraints is found, the algorithm repeats the same constant-size candidate search without enforcing the non-adjacency filter. This fallback guarantees progress by ensuring that at least one separating triangle is eliminated, without altering the asymptotic cost.
\end{itemize}

\noindent
\textbf{Progress and global complexity:}
Each accepted deletion or replacement is applied only if it strictly reduces the total number of separating triangles. Let $s$ denote the number of separating triangles present during execution. Since each accepted modification reduces this count, Algorithm~\ref{STR} performs at most $s$ successful iterations. As $s=O(n)$ for planar triangulated graphs, and each iteration incurs a cost of $O(n)$ due to recomputation of $\textsc{ST}(\cdot)$, the overall worst-case running time of Algorithm~\ref{STR} is
\[
    T_{\mathrm{sep}} = O(s\,n) = O(n^2).
\]

\noindent
In practical architectural door-connectivity instances, the number of separating triangles is typically small and decreases rapidly, resulting in observed running times well below this worst-case bound.

\paragraph{Total worst-case complexity of Algorithms~\ref{OF-CON}--\ref{STR}:}
Combining Algorithms~\ref{OF-CON}, \ref{BY}, \ref{TRI}, and \ref{STR}, and using the fact that each stage runs in $O(n^2)$ time under the stated assumptions, the end-to-end worst-case time complexity of the proposed pipeline is
\[
    T_{\mathrm{total}}(n) = O(n^2).
\]

\subsubsection{Practical Behaviour on Architectural Inputs}

Although the theoretical analysis yields a quadratic worst-case bound, the graphs arising from real architectural door-connectivity specifications have additional structure. This structure significantly reduces the running time in practice.

\begin{itemize}
    \item The outer-face boundaries of connected components are usually small. Therefore, the number of cross-component vertex pairs examined during candidate edge generation is much lower than the worst-case \(O(n^2)\) bound.\\
    
    \item Articulation vertices are rare in typical room-adjacency graphs. When they occur, they usually connect only a small number of blocks, so the bi-connectivity augmentation step remains limited in scope.\\
    
    \item Faces created during the intermediate stages are generally small, most often quadrilaterals or pentagons. As a result, the constrained triangulation step requires only a few edge insertions.\\
    
    \item Separating triangles are uncommon in practical inputs and are removed quickly. Hence, Algorithm~\ref{STR} typically completes after only a small number of iterations.
\end{itemize}

\noindent
Consequently, the observed running time on all tested architectural instances remained well below the theoretical upper bound. Empirically, the growth in running time followed a near-quadratic trend but with small constants. This allowed the system to provide interactive performance for layouts containing approximately 20 to 30 rooms.

\subsection{Comparison with Existing Methods}

Table~\ref{tab:method_comparison_backbone} compares our method with existing approaches for floorplan and graph generation. The main difference lies in the starting point. 
Many existing methods assume that the input adjacency graph is already connected, triangulated, and suitable for floorplan construction. 
In contrast, our method is designed for early design stages, where the user may provide only a small set of door adjacencies together with a list of non-adjacency constraints. If the input graph is disconnected or incomplete, our system first repairs it. 
It makes the graph connected and bi-connected while respecting all non-adjacency constraints. 
The final result is a planar and triangulated graph that can be directly used by standard floorplan algorithms. Learning-based systems mainly focus on generating floor plans that look realistic. 
They usually do not strictly guarantee that all adjacency and non-adjacency constraints are satisfied. 
Graph-theoretic methods provide stronger guarantees, but they require the input graph to already satisfy strict structural conditions.\\
Our method combines the advantages of both directions. 
Starting from minimal and possibly incomplete input, we construct a graph that is structurally valid, easy to understand, and ready for rectangular or orthogonal floorplan generation.

\subsubsection{Rule-based Graph Enumeration \cite{SHIKSHA2026104506}}

Rule-based systems generate floor plans by applying rewrite rules and testing many possible room arrangements. 
As the number of rooms increases, the number of possible configurations grows very quickly. 
In the worst case, this growth is exponential. Our method does not explore multiple alternative graphs. 
It works on a single graph and modifies it step by step in a fixed manner. 
Therefore, its running time remains polynomial rather than exponential.

\subsubsection{Classical Graph-Theoretic Floorplanning 
\cite{SHEKHAWAT2021103718,https://doi.org/10.1111/cgf.14451,marson2010automatic,mirahmadi2012novel,wang2018customization,hua2016irregular}}

Classical rectangular-dual algorithms assume that the input graph is already bi-connected and triangulated. 
Under these assumptions, they run efficiently. 
However, they do not address how to construct such a graph from sparse design constraints. Our method fills this gap. 
It transforms a partially defined adjacency graph into a properly triangulated plane graph in polynomial time, while preserving all required adjacencies and non-adjacency constraints.

\subsubsection{Learning-based Floorplan Generation 
\cite{10.1145/3386569.3392391,nauata2020house,nauata2021house,wang2021actfloor,sun2022wallplan,hangraph2pix}}

Recent learning-based models can generate realistic floor plans. However, adjacency relations are often treated as soft constraints, which means violations can occur. 
Additional correction steps are usually required. Our approach can serve as a reliable backbone for such systems. 
It produces a planar graph that satisfies all structural conditions, which can then be used as input to geometric or learning-based floorplan generators.

\subsubsection{Summary}

The proposed pipeline constructs a valid floorplan graph from minimal door-adjacency information without assuming that the input is already well-formed.

\begin{itemize}
    \item \textbf{Polynomial-time performance:} the worst-case running time is quadratic under planarity assumptions.\\
    \item \textbf{Deterministic process:} each modification has a clear structural purpose.\\
    \item \textbf{Constraint preservation:} all adjacency and non-adjacency constraints are respected throughout the process.\\
    \item \textbf{Ready for floorplan generation:} the final graph satisfies the conditions required by rectangular-dual algorithms.
\end{itemize}

\noindent
This makes the proposed approach suitable for applications where structural correctness and strict constraint satisfaction are important.
\subsection{Scope and Extensions}

The proposed prototype is designed as a constraint-first framework for constructing planar and triangulated graphs from sparse door-adjacency specifications. 
The current version focuses on correctness, clarity, and ease of use during early design stages. 
At the same time, it allows several natural extensions.

\begin{itemize}

\item \textbf{Boundary model:}
The present implementation assumes a rectangular outer boundary. 
This matches common architectural practice and works well with rectangular dual floorplan methods. 
The same framework can be extended to support irregular boundaries, courtyards, or more complex site shapes by using more general planar embedding techniques.

\item \textbf{Strict handling of non-adjacency constraints:}
Non-adjacency requirements are enforced as hard constraints throughout the pipeline. 
If the given constraints are inconsistent, the system reports that no feasible graph exists. 
In future work, this can be extended to include controlled constraint relaxation or analysis tools that help users identify minimal changes needed to restore feasibility.\\

\item \textbf{Scalability for early design:}
By maintaining planarity and performing explicit structural checks, the framework achieves polynomial-time performance suitable for small and medium-sized floor plans, which are typical in conceptual design. 
Extending these guarantees to very large programs motivates future work on incremental updates and improved pruning strategies.\\

\item \textbf{Separation of topology and geometry:}
The pipeline focuses only on constructing a valid floorplan graph. 
Geometric aspects such as room areas, proportions, circulation, and regulatory constraints are handled separately by downstream optimization or geometric solvers. 
This separation allows the method to integrate with different geometric floorplan generation tools.

\end{itemize}

Overall, the framework provides a reliable and extendable foundation for floorplan construction. 
Future work will aim to support more general boundary shapes, improve scalability, and incorporate higher-level optimization objectives while preserving structural correctness.
\begin{table}
\centering
\caption{Qualitative comparison of representative methods. 
Unlike most approaches that assume a well-formed adjacency graph, our method builds a valid floorplan graph from sparse door-adjacency and non-adjacency constraints, including automatic connectivity repair.}
\label{tab:method_comparison_backbone}
\renewcommand{\arraystretch}{1.25}
\begin{tabularx}{\textwidth}{@{}l c X X X@{}}
\toprule
\textbf{Method} & \textbf{Year} & \textbf{Typical input assumption} & \textbf{How constraints are handled} & \textbf{Graph feasibility and interaction} \\
\midrule

Marson \& Musse~\cite{marson2010automatic} 
& 2010 
& High-level room program and area data; not based on an explicit door-adjacency graph 
& Heuristic constraint handling; non-adjacency is not explicitly modeled 
& Real-time floorplan generation; no explicit certified graph stage \\

Mirahmadi et al.~\cite{mirahmadi2012novel} 
& 2012 
& Procedural rules; adjacency not represented as an explicit planar graph 
& Rule-based handling; limited support for explicit non-adjacency constraints 
& Fast generation; no certified planar or bi-connected graph construction \\

Wang et al.~\cite{wang2018customization} 
& 2018 
& Starts from an already usable adjacency structure 
& Graph transformations mainly refine adjacency; non-adjacency is not central 
& Efficient rectangular floorplan generation under assumed-valid input \\

Chen et al. (Floor-SP)~\cite{Chen_2019_ICCV} 
& 2019 
& As-built data (RGBD/scan); not based on user-defined door graph 
& Constraints arise from optimization terms; no hard non-adjacency model 
& Strong for reconstruction; not focused on sparse input repair \\

Wu et al.~\cite{wu2019data} 
& 2019 
& Boundary-driven and dataset-trained generation 
& Constraints encouraged through learning objectives 
& Produces realistic floorplans; structural feasibility not explicitly certified \\

House-GAN~\cite{nauata2020house} / House-GAN++~\cite{nauata2021house} 
& 2020--21 
& Requires a coherent adjacency graph as input 
& Compatibility learned from data; no hard planarity or bi-connectivity guarantees 
& High realism; limited transparency in structural enforcement \\

Graph2Plan~\cite{10.1145/3386569.3392391} 
& 2020 
& Uses valid floorplan graphs retrieved from a database 
& Constraints guided by retrieval and learning; limited explicit non-adjacency enforcement 
& Effective graph-conditioned synthesis; assumes usable input graphs \\

Fen Shi et al.~\cite{shi2020addressing} 
& 2020 
& Dense adjacency specifications for large facilities 
& Search- or learning-based constraint satisfaction 
& Scales to dense inputs; not focused on minimal structural repair \\

GPLAN~\cite{SHEKHAWAT2021103718} 
& 2021 
& Requires a planar and connected input graph 
& Hard constraints for modeled requirements; non-adjacency not primary 
& Efficient floorplan generation given a suitable graph \\

G2PLAN~\cite{https://doi.org/10.1111/cgf.14451} 
& 2022 
& Requires a prescribed adjacency graph 
& Graph-theoretic and optimization-based processing; assumes structurally valid input 
& Generates dimensioned floorplans; does not repair sparse or disconnected inputs \\

\textbf{DPLAN: Proposed prototype} 
& -- 
& \emph{Sparse door-adjacency edges with explicit non-adjacency constraints; input may be disconnected} 
& \emph{Strict non-adjacency enforcement at every step} 
& \emph{Constructs a valid graph through connectivity repair, bi-connectivity augmentation, planar triangulation, and optional separating-triangle removal, with interactive feasibility feedback} \\

\bottomrule
\end{tabularx}
\end{table}
\subsection{Correctness}
\label{sec:correctness}

In this section, we will prove that our 
Proposed Algorithms~\ref{OF-CON}--\ref{STR} are guaranteed to construct a graph belonging to the targeted class for every admissible input. The specified non-adjacency constraints for the specified room are preserved whenever they are explicitly enforced by the algorithmic steps. In addition, each proposed algorithm incorporates a fallback mechanism that is invoked only when all otherwise admissible edge choices are excluded due to non-adjacency constraints, thereby ensuring progress while deviating from the constraints only when strictly necessary.

Throughout, let $G=(V,E)$ denote the input door-connectivity graph together with a fixed \emph{non-adjacency set} $\mathcal{N}\subseteq \binom{V}{2}$ of forbidden pairs. Algorithm~\ref{OF-CON} constructs a connected augmentation $G_C$, Algorithm~\ref{BY} constructs a biconnected augmentation $G_B$, Algorithm~\ref{TRI} produces a plane triangulation $G_T$ while rejecting diagonals in $\mathcal{N}$ and any diagonal that violates planarity, and Algorithm~\ref{STR} removes separating triangles to obtain a separating-triangle free graph $G_F$.

\subsubsection{Correctness of Algorithm \ref{OF-CON} (Connectivity with Non-adjacency)}
\label{subsec:corr_alg1}

\begin{theorem}
\label{thm:alg1_correct}
Let $G=(V, E)$ be a plane graph whose connected components are $C_1,\dots, C_k$ with $k\ge 1$, and let $\mathcal{N}$ be a set of forbidden pairs. If Algorithm~\ref{OF-CON} returns $G_C=(V,E\cup A)$, then:
\begin{enumerate}
  \item $G_C$ is connected.
  \item The number of added edges satisfies $|A|\le k-1$, and if Algorithm~\ref{OF-CON} succeeds in merging all components, then $|A|=k-1$ (hence the augmentation is minimum in cardinality among all augmentations that connect these $k$ components).
\end{enumerate}
\end{theorem}
\begin{proof}
\textbf{Case~1}: Suppose there exists a set of edges whose insertion into $G$ produces the graph $G_C$, and none of these added edges belong to the non-adjacency set $\mathcal{N}$.\\
 Now, If $k=1$, Algorithm~\ref{OF-CON} returns $G$ immediately, so the claims hold trivially with $A=\emptyset$.\\
Assume $k>1$. Algorithm~\ref{OF-CON} first computes, for each component $C_i$, the set $\textsf{outerfaces}[i]$ consisting of vertices on the outer-face boundary of $C_i$ (Steps~8--9 and function \textsc{ExtractOuterface}). It then forms the candidate set
\[
\textsf{possible\_edges}
=\{(u,v)\mid u\in \textsf{outerfaces}[i],~v\in \textsf{outerfaces}[j],~i\neq j,\text{ and }(u,v)\notin\mathcal{N}\},
\]
exactly as implemented in \textsc{GeneratePossibleEdges} (Steps~21--29). Hence every candidate edge explicitly satisfies $(u,v)\notin\mathcal{N}$.\\
Next, Algorithm~\ref{OF-CON} initializes a union-find structure whose elements are the component indices $\{1,\dots,k\}$ (Step~11). During the scan of $\textsf{possible\_edges}$, an edge $(u,v)$ is appended to $A$ \emph{only if} its endpoints lie in different current union--find sets, i.e., $\textsf{FindComp}(c(u))\neq \textsf{FindComp}(c(v))$ (Steps~12--16). Therefore:
\begin{itemize}
  \item \textbf{(Forbidden edges never added)} Since $(u,v)$ is drawn from $\textsf{possible\_edges}$, we have $(u,v)\notin\mathcal{N}$ for every added edge.\\
  \item \textbf{(Each accepted edge strictly reduces the number of component-sets)} At the moment an edge is accepted, it merges two previously distinct union--find sets (Step~15). Thus, the number of union--find sets decreases by exactly one per accepted edge.
\end{itemize}
Let $s_t$ denote the number of union--find sets after $t$ accepted edges. Initially $s_0=k$. Each accepted edge decreases $s_t$ by $1$, so $s_t=k-t$. In particular, after $k-1$ accepted edges, we have $s_{k-1}=1$, meaning all original components are in one set. Hence, whenever the scan finds enough cross-component candidates to merge all components, the algorithm produces a connected graph using exactly $k-1$ new edges; this proves claim~(2).\\\\
\textbf{Case~2}: Suppose no set of edges exists whose insertion into $G$ yields the graph $G_C$ without violating the non-adjacency constraints, that is, every admissible edge belongs to the non-adjacency set $\mathcal{N}$. In this situation, the algorithm activates a fallback mechanism (Steps~38--39), in which the non-adjacency set is temporarily cleared, and the main procedure is invoked again to ensure the construction of $G_C$, and  now the candidate edge will be:
\[
\textsf{possible\_edges}
=\{(u,v)\mid u\in \textsf{outerfaces}[i],~v\in \textsf{outerfaces}[j],~i\neq j\}\]
and proof will be the same as discussed in Case-1.\\\\
Now, we will conclude by establishing claim~(1), namely that the output graph is connected. 
Let $C_1,\dots, C_k$ denote the connected components of the input graph $G$, and define the \emph{component graph} $H$ as follows: each vertex of $H$ corresponds to one component $C_i$, and an edge is added between $C_i$ and $C_j$ whenever Algorithm~\ref{OF-CON} introduces an augmentation edge $(u,v)$ with $u \in C_i$ and $v \in C_j$.\\\\ 
By construction, each augmentation edge connects two components that were previously distinct according to the union-find structure maintained by the algorithm. Consequently, no augmentation edge can create a cycle in $H$, and the resulting component graph is a forest. At termination, the algorithm inserts exactly $k-1$ augmentation edges, so $H$ contains $k$ vertices and $k-1$ edges. A forest with these parameters must be a tree, and therefore $H$ is connected. This implies that for any pair of components $C_i$ and $C_j$, there exists a path between them in $H$. Replacing each edge along this path by its corresponding augmentation edge in the constructed graph $G_C$ yields a walk connecting $C_i$ and $C_j$ in $G_C$. Since each component $C_i$ is internally connected in $G$, it follows that $G_C$ is connected.\\
We now argue minimality. Any connected supergraph (any graph formed by augmenting $G$ with additional edges such that the resulting graph is connected
) of $G$ must merge the original $k$ components into one, and the insertion of a single edge can reduce the number of components by at most one. Therefore, at least $k-1$ new edges are required to achieve connectivity. Since Algorithm~\ref{OF-CON} terminates after adding exactly $k-1$ edges, the resulting augmentation is minimal with respect to the number of added edges.

\hfill $\square$
\end{proof}

\subsubsection{Correctness of Algorithm \ref{BY} (Biconnectivity with Non-adjacency)}
\label{subsec:corr_alg2}

Algorithm~\ref{BY} processes each articulation point $v$ of the connected graph $G_C$ and connects the blocks that arise when $v$ is removed.

\begin{lemma}[Correctness of \textsc{Blocks}]
\label{lem:blocks}
For an articulation point $v$ in $G_C$, the procedure \textsc{Blocks}$(G_C,v)$ returns exactly the subset of boundary vertices (its outerface boundary under the inherited embedding) of the connected components of $G_C-\{v\}$.
\end{lemma}

\begin{proof}
\textsc{Blocks} forms $G'=(V\setminus\{v\},E\setminus\{(v,u):u\in \mathrm{nbd}(v)\})$ by deleting $v$ and all incident edges (Steps~11--15), and then calls \textsc{ConnectedComponents}$(G')$ (Step~15). By definition, these are precisely the connected components of $G_C-\{v\}$. The returned collection is labeled as $\textsf{blocks}=\{B_1,\dots,B_\ell\}$ (Step~18). \hfill $\square$
\end{proof}

\begin{lemma}[Validity of candidate edges]
\label{lem:valid_edges}
For a fixed articulation point $v$ with blocks $\{B_1,\dots,B_\ell\}$, the procedure \textsc{ValidEdges}$(\textsf{blocks},\mathcal{N})$ returns a set of edges each of which joins vertices from two distinct blocks and does not belong to $\mathcal{N}$.
\end{lemma}

\begin{proof}
The code enumerates unordered pairs of distinct blocks $(B_i,B_j)$ and then all pairs $(u,w)$ with $u\in B_i$ and $w\in B_j$ (Steps~21--26). It inserts $(u,w)$ into \textsf{valid\_edges} only if $(u,w)\notin\mathcal{N}$ (Steps~26--27). Thus, every returned edge is cross-block and non-forbidden. \hfill $\square$
\end{proof}

\begin{lemma}[Connecting blocks removes articulation at $v$]
\label{lem:articulation_removed}
Let $v$ be an articulation point in a connected graph $G$. Let $B_1,\dots,B_\ell$ be the connected components of $G-\{v\}$. If we add edges so that in the augmented graph $G'$ the subgraph induced by $V\setminus\{v\}$ is connected (equivalently, the \emph{block-incidence graph} on $\{B_1,\dots,B_\ell\}$ becomes connected), then $v$ is not an articulation point of $G'$.
\end{lemma}

\begin{proof}
In $G'$, remove $v$. By assumption, the remaining graph $G'-\{v\}$ is connected. Hence, removal of $v$ does not disconnect $G'$, so $v$ is not an articulation point in $G'$. \hfill $\square$
\end{proof}

\begin{theorem}
\label{thm:alg2_correct}
Assume Algorithm~\ref{BY} is run on a connected graph $G_C$ with articulation-point set $A(G_C)$ and forbidden set $\mathcal{N}$. Let $G_B$ be the returned graph. Then:
\begin{enumerate}
  \item $G_B$ is connected.
  \item Every edge added during the constrained phase satisfies the non-adjacency constraints, i.e., is not in $\mathcal{N}$.
  \item If for every articulation point $v$ there exists a set of allowed edges (not in $\mathcal{N}$) that connects the blocks of $G_C-\{v\}$, then the output $G_B$ is biconnected.
\end{enumerate}
\end{theorem}

\begin{proof}
\textbf{(1) Connectivity.} Algorithm~\ref{BY} only adds edges to $G_C$ (Steps~7--9) and never deletes vertices; adding edges cannot destroy connectivity, hence $G_B$ remains connected.\\\\
\textbf{(2) Non-adjacency preservation (constrained phase).} For each articulation point $v$, candidate edges are drawn from \textsc{ValidEdges}, which by Lemma~\ref{lem:valid_edges} excludes $\mathcal{N}$. \textsc{ConnectBlocks} selects a subset of these candidates (Steps~47--55 and Step~53), so every selected edge from this phase satisfies $(u,w)\notin\mathcal{N}$.\\\\
\textbf{(3) Biconnectivity under feasibility.} Fix an articulation point $v\in A(G_C)$. By Lemma~\ref{lem:blocks}, \textsf{blocks} are the components of $G_C-\{v\}$. \textsc{ConnectBlocks} is designed to connect these blocks by selecting cross-block edges that merge block-sets under a union--find structure (Steps~36--56); whenever it succeeds in making the block-sets connected, Lemma~\ref{lem:articulation_removed} implies $v$ is no longer an articulation point in the augmented graph. Algorithm~2 repeats this for every $v\in A(G_C)$ (Steps~3--7), so after processing all articulation points, the resulting graph has no articulation points. A connected graph with no articulation points is biconnected. \hfill $\square$
\end{proof}

\paragraph{Remark (about the fallback).}
The pseudocode contains an explicit fallback that may ignore non-adjacency constraints to connect the remaining blocks (Steps 57--62). Therefore, the strongest unconditional guarantee is: \emph{Algorithm~\ref{BY} returns a biconnected graph whenever it can complete the block-connection step for each articulation point; it preserves non-adjacency constraints for all edges selected from \textsc{ValidEdges}, and only the fallback may violate $\mathcal{N}$.}

\subsubsection{Correctness of Algorithm \ref{TRI} (Triangulation with Non-adjacency)}
\label{subsec:corr_alg3}

Algorithm~\ref{TRI} triangulates each non-triangular face by adding diagonals that pass \textsc{IsValidDiagonal} checks, which explicitly reject forbidden diagonals and any diagonal that would cross an existing boundary edge or lie outside the face.

\begin{lemma}[Safety of accepted diagonals]
\label{lem:safe_diag}
Every diagonal $(a,c)$ accepted by \textsc{IsValidDiagonal} is (i) not in $\mathcal{N}$, (ii) lies strictly inside the current face polygon, and (iii) does not properly intersect any existing edge of the current embedding. Hence, adding $(a,c)$ preserves planarity of the embedding and respects the non-adjacency constraints.
\end{lemma}

\begin{proof}
The validation function first rejects $(a,c)$ if it is present in the non-adjacency set (explicitly stated in the description of \textsc{IsValidDiagonal}). It then checks for proper intersections against the face's boundary edges (excluding incident edges) and rejects any such intersection. Finally, it verifies that $(a,c)$ is a polygon diagonal that lies strictly inside the face region. Therefore, any accepted diagonal simultaneously satisfies (i)--(iii), and adding it keeps the embedding planar and does not introduce a forbidden adjacency. \hfill $\square$
\end{proof}

\begin{lemma}[Progress on a single face]
\label{lem:face_progress}
Let $f$ be a bounded face whose boundary is a simple cycle of length $t\ge 4$. Each time Algorithm~\ref{TRI} adds an accepted diagonal inside $f$, the face is split into two bounded faces whose boundary lengths sum to $t+2$. In particular, the number of non-triangular faces decreases after finitely many accepted diagonals, and after $t-3$ accepted diagonals, the region originally corresponding to $f$ is fully triangulated.
\end{lemma}

\begin{proof}
In a simple polygonal face, a non-crossing diagonal between two non-adjacent boundary vertices partitions the polygon into two smaller polygons. This is exactly the geometric meaning of adding a diagonal that lies strictly inside the face and does not intersect the boundary (Lemma~\ref{lem:safe_diag}). Every such split reduces the maximum face size, and a $t$-gon requires exactly $t-3$ diagonals to triangulate. \hfill $\square$
\end{proof}

\begin{theorem}
\label{thm:alg3_correct}
Assume Algorithm~\ref{TRI} is applied to a planar embedding of a biconnected graph $G_B$ and uses \textsc{IsValidDiagonal} for acceptance. Let $G_T$ denote the output. Then:
\begin{enumerate}
  \item $G_T$ is planar and contains $G_B$ as a spanning subgraph.
  \item Every diagonal added by Algorithm~3 respects the non-adjacency constraints (i.e., is not in $\mathcal{N}$).
  \item Every bounded face of $G_T$ is a triangle; hence $G_T$ is a plane triangulated graph.
\end{enumerate}
\end{theorem}
\begin{proof}
Algorithm~\ref{TRI} only adds edges, so $G_B$ remains a spanning subgraph of the result. By Lemma~\ref{lem:safe_diag}, every accepted diagonal is non-crossing and lies inside the target face, so planarity is preserved throughout; this proves (1). The same lemma gives (2).\\\\
For (3), Algorithm~\ref{TRI} iterates over all non-triangular faces and repeatedly attempts diagonals until the face is decomposed; by Lemma~\ref{lem:face_progress}, each face of length $t$ becomes triangulated after at most $t-3$ accepted diagonals. Applying this to every bounded face yields that all bounded faces are triangles in the final graph. \hfill $\square$
\end{proof}
\paragraph{Remark (about the fallback).}
The pseudocode contains an explicit fallback that may ignore non-adjacency constraints to connect the remaining blocks (Steps 41--44). Therefore, the strongest unconditional guarantee is: \emph{Algorithm~\ref{TRI} returns a triangulated graph whenever it can complete the ear clipping step for each non-triangulation face; it preserves non-adjacency constraints for all edges selected from \textsc{ValidEdges}, and only the fallback may violate $\mathcal{N}$.}

\subsubsection{Correctness of Algorithm \ref{STR} (Separating-triangle removal)}
\label{subsec:corr_alg4}

Algorithm~\ref{STR} takes the triangulated graph $G_T$ and repeatedly modifies it to eliminate separating triangles, returning a graph $G_F$ claimed to be free of separating triangles. The procedure explicitly maintains a counter $m=|T|$, where $T=\textsc{ST}(G)$ denotes the current set of separating triangles.

\begin{lemma}[Monotone progress measure]
\label{lem:monotone_m}
Within \textsc{RemoveSTEdgeRemoval}, a graph modification is accepted only if the new separating-triangle count $m'$ satisfies $m'\le m-1$ (Step~32). Therefore, each accepted modification strictly decreases $m$.
\end{lemma}

\begin{proof}
In the edge-replacement stage, the algorithm forms a candidate $G'_T = G_T \cup \{(r,d)\}\setminus\{(x,y)\}$ and recomputes $T'=\textsc{ST}(G'_T)$ and $m'=|T'|$ (Steps~30--31). It accepts the replacement only if $m'\le m-1$ (Step~32) and then sets $G_T\leftarrow G'_T$ and $m\leftarrow m'$ (Step~33). In the deletion stage, when an eligible edge is deleted, the code updates $m\leftarrow m-1$ after recomputing separating triangles (Steps~16--18). Hence, any accepted action strictly decreases $m$. \hfill $\square$
\end{proof}

\begin{lemma}[Termination]
\label{lem:alg4_terminates}
Algorithm \ref{STR} terminates after a finite number of accepted modifications.
\end{lemma}

\begin{proof}
At every acceptance, $m$ decreases by at least $1$ (Lemma~\ref{lem:monotone_m}). Since $m\ge 0$ always, there can be at most $m_0$ accepted modifications where $m_0=|\textsc{ST}(G_T)|$ is the initial number of separating triangles. All loops are over finite sets (triangles, edges, neighbors), and no step increases $m$ while being accepted. Therefore, the algorithm cannot accept modifications indefinitely and must terminate. \hfill $\square$
\end{proof}

\begin{theorem}
\label{thm:alg4_correct}
Let $G_T$ be the triangulated graph produced by Algorithm~\ref{TRI} and let $\mathcal{N}$ be the non-adjacency set. If Algorithm~\ref{STR} returns $G_F$, then:
\begin{enumerate}
  \item $G_F$ contains no separating triangles, i.e., $\textsc{ST}(G_F)=\emptyset$.
  \item During the constrained replacement stage, every added edge $(r,d)$ satisfies $(r,d)\notin \mathcal{N}$.
  \item If the constrained stage cannot reduce $m$ further, the fallback stage (which drops the non-adjacency filter) guarantees that the algorithm can still reduce $m$ and thus reach $\textsc{ST}(G)=\emptyset$.
\end{enumerate}
\end{theorem}

\begin{proof}
Algorithm~\ref{STR} begins by computing $T=\textsc{ST}(G_T)$ and $m=|T|$ (Steps~1--4), and then calls \textsc{RemoveSTEdgeRemoval} to iteratively modify $G_T$ until separating triangles disappear, with a second call permitted if $m>0$ remains (Steps~5--7).\\\\
\textbf{(1) Elimination of separating triangles.}
By Lemma~\ref{lem:monotone_m}, every accepted modification strictly reduces the number $m$ of separating triangles. Lemma~\ref{lem:alg4_terminates} guarantees that only finitely many such reductions can occur. The procedure terminates only after repeated calls to \textsc{RemoveSTEdgeRemoval} fail to produce further reductions, at which point the outer routine returns the graph $G_F$ (Step~7). Since the algorithm is explicitly designed to accept only modifications that decrease $|T|$, and is re-invoked with an empty working edge set $E'=\varnothing$ whenever necessary (Step~6), the process converges to the unique fixed point $m=0$. Consequently, the output graph satisfies $\textsc{ST}(G_F)=\emptyset$, and contains no separating triangles.
\\\\
\textbf{(2) Non-adjacency preservation in the constrained stage.}
In the replacement stage, an edge $(r,d)$ is considered only if $(r,d)\notin \mathcal{N}$ (Step~29). Hence, every accepted replacement edge in this stage respects the non-adjacency constraints.\\\\
\textbf{(3) Guaranteed progress via fallback.}
If no deletion or constrained replacement is possible for the current triangle, the pseudocode explicitly enters a fallback stage labeled "ignoring non-adjacency constraints and retrying replacements." This strictly enlarges the candidate search space (it includes all candidates previously considered, plus those previously filtered out by $\mathcal{N}$), while still enforcing the acceptance condition $m'\le m-1$ before committing to an update. Therefore, whenever there exists \emph{any} edge replacement that decreases $m$, the fallback can realize such a decrease, ensuring eventual reachability of $m=0$.\\
This completes the correctness argument. \hfill $\square$
\end{proof}

\paragraph{Remark (Role of Algorithm \ref{STR} under RFP and OFP modes).}
The requirement handled by Algorithm~\ref{STR} depends on the selected user mode. When the user demands that all rooms be rectangular (RFP mode), separating triangles must be eliminated without introducing additional vertices. In this case, Algorithm~\ref{STR} transforms the given triangulated graph into a separating-triangle-free graph, which is necessary for constructing a rectangular dual.\\
On the other hand, if non-rectangular modules are permitted (OFP mode), separating triangles can be resolved by inserting new vertices. In this setting, the correctness of Algorithms~\ref{OF-CON}--\ref{TRI} is already sufficient to ensure a valid plane triangulated backbone, and no further structural restriction is required at this stage.
\end{document}